\documentclass[lettersize,journal]{IEEEtran}
\usepackage{amsmath,amssymb,amsfonts}
\usepackage{graphicx,color}
\usepackage{textcomp}
\usepackage{xcolor}
\usepackage{hyperref}
\usepackage{algorithm}
\usepackage{amsmath,bm,mathtools,mathrsfs,mathabx}
\usepackage{multicol}
\usepackage{amsthm}
\usepackage{float}
\usepackage{comment}

\usepackage{cleveref}
\usepackage{cite}
\usepackage{enumitem}
\usepackage{outlines}
\usepackage{algpseudocode}
\AtBeginDocument{\definecolor{tmlcncolor}{cmyk}{0.93,0.59,0.15,0.02}\definecolor{NavyBlue}{RGB}{0,86,125}}
\newcommand\scalemath[2]{\scalebox{#1}{\mbox{\ensuremath{\displaystyle #2}}}}

\newtheorem{theorem}{Theorem}

\newtheorem{lemma}{Lemma}

\newtheorem{assumption}{Assumption}

\newtheorem{remark}{Remark}
\newtheorem{definition}{Definition}

\newcommand*{\hermtr}{{\mathsf{H}}}

\begin{document}

\title{Data-Driven Learning of Two-Stage Beamformers in Passive IRS-Assisted Systems with Inexact Oracles}

\author{Spyridon Pougkakiotis, Hassaan Hashmi,~\IEEEmembership{Student Member, IEEE}, Dionysis Kalogerias,~\IEEEmembership{Senior Member, IEEE \vspace{-6pt}}
        % <-this % stops a space
\thanks{S. Pougkakiotis is with the Department of Mathematics, King's College London, UK (email: spyridon.pougkakiotis@kcl.ac.uk). H. Hashmi and D. Kalogerias are with the Department of Electrical and Computer Engineering, Yale University, New Haven, CT, USA (email: \{hassaan.hashmi,
dionysis.kalogerias\}@yale.edu).}
\thanks{This work is supported by the US NSF under Grant 2242215.}}% <-this % stops a space
%\thanks{Manuscript received April 19, 2021; revised August 16, 2021.}}



%\receiveddate{XX Month, XXXX}
%\reviseddate{XX Month, XXXX}
%\accepteddate{XX Month, XXXX}
%\publisheddate{XX Month, XXXX}
%\currentdate{XX Month, XXXX}
%\doiinfo{XXXX.2022.1234567}
%\allowdisplaybreaks

%\markboth{}{Author {et al.}}

%\title{Data-Driven Learning of Two-Stage passive IRS-assisted beamformers with inexact oracles}
%\title{Data-Driven Learning of Two-Stage Beamformers in Passive IRS-Assisted Systems with Inexact Oracles}

%\author{Spyridon Pougkakiotis\authorrefmark{1}, Hassaan Hashmi\authorrefmark{2}, Student Member, IEEE, \\and Dionysis Kalogerias\authorrefmark{2}, Senior Member, IEEE}
%\affil{Department of Mathematics, King's College London, England, UK}
%\affil{Department of Electrical and Computer Engineering, Yale University, CT, USA}
%\corresp{Corresponding author: Spyridon Pougkakiotis (email: spyridon.pougkakiotis@kcl.ac.uk).}
%\authornote{This paragraph of the first footnote will contain support information, including sponsor and financial support acknowledgment. For example, ``This work was supported in part by the U.S. Department of Commerce under Grant 123456.''}
\maketitle

\begin{abstract}
We develop an efficient data-driven and model-free unsupervised learning algorithm for achieving fully passive intelligent reflective surface (IRS)-assisted optimal short/long-term beamforming in wireless communication networks. The proposed algorithm is based on a zeroth-order stochastic gradient ascent methodology, suitable for tackling two-stage stochastic nonconvex optimization problems with continuous uncertainty and unknown (or ``black-box'') terms present in the objective function, via the utilization of \textit{inexact evaluation oracles}.
We showcase that the algorithm can operate under realistic and general assumptions, and establish its convergence rate close to some stationary point of the associated two-stage (i.e., short/long-term) problem, particularly in cases where the second-stage (i.e., short-term) beamforming problem (e.g., transmit precoding) is solved inexactly using an arbitrary (inexact) oracle. 
The proposed algorithm is applicable on a wide variety of IRS-assisted optimal beamforming settings, while also being able to operate without (cascaded) channel model assumptions or knowledge of channel statistics, and over arbitrary IRS physical configurations; thus, no active sensing capability at the IRS(s) is needed.
%
%Apart from its generality, the developed theory paves the way for the creation of highly efficient algorithms that are deployable in realistic beamforming scenarios, and can be utilized to solve challenging (large-scale) instances  that are out-of-reach of alternative state-of-the-art schemes. 
%
%By characterizing the propagation of errors arising from the utilization of practical inexact oracles for the short-term (or second-stage) beamforming, we are able to design highly efficient algorithmic strategies that are 
Our algorithm is numerically demonstrated to be very effective in a range of experiments pertaining to a well-studied MISO downlink model, including scenarios demanding physical IRS tuning (e.g., directly through varactor capacitances), even in large-scale regimes.
\vspace{-3pt}
\end{abstract}



\vspace{-6pt}
\section{Introduction}
\IEEEPARstart{W}{ireless} communication systems in the current era of massive information necessitate the development and deployment of efficient and reliable \emph{beamforming strategies}, as a means to achieving effective signal alignment and/or interference cancellation, for improving certain metrics of \emph{Quality of Service (QoS)}, such as rate, signal-to-noise ratio, decoding performance, or spectral efficiency. As a result, several novel technologies have been thoroughly investigated in relatively recent literature \cite{mimo:larsson2014massive, udn:kamel2016ultradense, udn:gotsis2016ultradense, 5g:andrews2014will, 5g:shafi20175g}, reinforcing and improving traditional beamforming techniques \cite{yoo2005optimality, wmmseShi2011}, which typically struggle with so-called \emph{non-line-of-sight} (non-LOS) losses (e.g., in the case of highly directional mmWaves \cite{mmWave_challenges_opportunities}), energy consumption, poor scalability, or increased latency. 
\par One popular technology aimed at scalably improving conventional beamforming approaches, and which has been vigorously investigated over the last few years, is that of \emph{Intelligent Reflecting Surfaces} (IRSs or RISs) \cite{irs_design:sievenpiper2003}. An IRS is a metasurface comprised of a planar array of reflecting elements, whose amplitude and phase can be tuned via changing certain physical parameters (or knobs). IRSs are purposed to minimize non-LOS losses without the need of resorting to ultra-dense networks (e.g., as in \cite{5g:udn_valenzuela2018ultra}), and with minimal overhead in terms of interference and energy consumption. Specifically, \emph{passive} IRS deployment (i.e., the deployment of IRSs with no active sensing onboard) is of paramount importance for circumnavigating some of the core challenges faced by conventional beamforming, while also being practically feasible. 
\par The potentially far-reaching benefits of enhancing standard beamforming schemes (often based on utility maximization) via the deployment of ultimately passive IRS devices at a very manageable overhead have resulted in the widespread study of such IRS-assisted wireless communication systems. In fact, there is a rich literature on IRS-aided optimal beamforming, which focuses on the development of methods and algorithms able to reliably yield high-quality (approximate) solutions to the associated optimization problems.

% \par To that end, two general optimization models have been typically explored in the relevant literature. The \textit{first} model treats IRSs in a short-term scale, and the tuning of the IRS parameters is performed simultaneously with standard symbol beamforming (i.e., precoding); see, e.g., \cite{reactive:wu2019intelligent, reactive:wu2019towards, reactive:wang2020intelligent, reactive:wu2020joint}. This approach entails highly unrealistic operational assumptions (especially in passive-IRS set-ups) and would introduce substantial computational and energy overheads if deployed in a practical setting. \dk{Here put all limitations, also mentioning the limitations of ML DRL methods that follow this approach}

\par To that end, two general system models have been primarily explored in the relevant literature, resulting in methods of distinct character as far as beamforming optimization is concerned. The \textit{first model} treats IRS elements as short-term scale (i.e., reactive) beamformers, and their tuning/selection is performed simultaneously with standard beamforming (i.e., symbol transmit precoding). This model originally calls for the exploitation of elaborate channel models and detailed system structure (e.g., spatial network topology and IRS configurations), leading to the design of specialized solvers to find (near-)optimally tuned IRS elements, on the basis of instantaneous cascaded channel state information (CSI); see, e.g., \cite{reactive:wu2019intelligent, reactive:wu2019towards, reactive:wang2020intelligent, reactive:wu2020joint,sco:yang2021}. While generally effective in simulated environments, this approach entails highly unrealistic operational assumptions (especially in passive IRS setups) and would introduce substantial computational and resource overheads if deployed in a practical setting, in particular due to the need of very fine-grained, perpetual estimation of instantaneous CSI (also most often necessitating active on-IRS sensing), as well as resource-inefficient continuous IRS optimal control.

While subsequent approaches relying on either supervised learning (SL) \cite{ml:taha2021,ml:yang2021} or end-to-end deep reinforcement learning (DRL) \cite{drl:mismar2019dqn,drl:huang2020ddpg} promise to alleviate one or more of the aforementioned issues, they ultimately introduce new challenges. In fact, the utilization of elaborate function approximators ---e.g., deep neural networks, employed either for value function approximation, or as predictors---, on which both such approaches rely on, is well-known to often increase problem complexity (especially without expert domain knowledge), result to non-interpretable models, hinder robustness, frequently induce overfitting and lack of generalizability, and/or to being sensitive to hyperparameter tuning. Similar conclusions can be made in regard to another more recent line of work that develops methods employing function approximators as (IRS) beamforming policy parameterizations trained through ``unsupervised'' learning schemes \cite{group2:chen2023graph, group2:jiang2021learning, group2:li2023gnn, group2:lim2023graph, group2:lyu2024investigating, group2:mehrabian2024joint, group2:peng2024risnet,group2:song2020unsupervised,group2:yang2024graph, group2:yeh2024enhanced}; such methods are essentially model-based as well, since they rely on fine-grained cascaded CSI estimation, especially during training (despite the fact that inference can often be performed in a model-free fashion).

The \textit{second system model} investigated in the literature (and also in the rest of this paper) treats optimal tuning of the IRS elements in a long-term timeframe, while performing standard optimal beamforming (e.g., symbol precoding) in the short-term. While this is a significantly more realistic approach, it results in optimization problems which are natural instances of \emph{two-stage nonconvex stochastic programming}  (e.g., see \cite{sco:guo2020larsson,sco:zhao2020tts, sco:zhao2021qos, group1:cai2020reconfigurable, group1:cao2022two, group1:chen2022two, group1:chen2023irs, group1:gan2022multiple, group1:jin2023low, group1:wang2022two, group1:wang2023spatial, group1:wu2023two, group1:yang2022multi, group1:zhai2021beamforming, group1:zhang2023deep, group1:zhao2022secrecy, IEEE_tran_Hashmietal}), and particularly challenging to solve.
%\cite{sco:guo2020larsson,sco:zhao2020tts, sco:zhao2021qos, group1:cai2020reconfigurable, group1:cao2022two, group1:chen2022two, group1:chen2023irs, group1:gan2022multiple, group1:jin2023low, group1:wang2022two, group1:wang2023spatial, group1:wu2023two, group1:yang2022multi, group1:zhai2021beamforming, group1:zhang2023deep, group1:zhao2022secrecy}), 
\par Up until recently, a state-of-the-art approach for successfully tackling the aforementioned nonconvex two-stage stochastic problems has relied on the utilization of \emph{Stochastic Successive Convex Approximation} (SSCA) techniques \cite{SSCO_1, SSCO_2}, as well as certain simplifying (e.g., linear) assumptions on the underlying cascaded IRS CSI model and overall interference pattern of the underlying communication system (e.g., \cite{sco:guo2020larsson,sco:zhao2020tts, sco:zhao2021qos, group1:wu2023two,group1:chen2022two, group1:chen2023irs, group1:jin2023low, group1:zhao2022secrecy, group1:cai2020reconfigurable}). However, except for being dependent on models as well as explicit cascaded (statistical) CSI estimation, SSCA-based techniques require the formulation of convex surrogates replacing the original long-term (i.e., first-stage) problems (see, e.g., the seminal work \cite{sco:zhao2020tts}), leading to suboptimal performance on their original non-convex counterparts, in settings as standard as weighted sumrate maximization. Likewise, while several other relevant approaches have been recently proposed \cite{group1:zhang2023deep, group1:cao2022two, group1:wang2022two, group1:gan2022multiple, group1:yang2022multi, group1:wang2023spatial, group1:zhai2021beamforming}, their applicability is also plagued by the need for cascaded, instantaneous or statistical, CSI estimation, again presupposing the imposition of explicit model structure. 
% Alternatively, to avoid such restrictive channel model assumptions, researchers  have explored model-free optimization strategies relying on supervised machine learning \cite{ml:taha2021,ml:yang2021} or deep reinforcement learning \cite{drl:mismar2019dqn,drl:huang2020ddpg}. Still, the former approach results in highly specialized, suboptimal and sensitive to network/channel behavior solutions 
% %schemes (due to surrogate problem formulations) 
% while, the latter approach relies on heuristically selected function approximators which, in the absence of domain knowledge, results in an increase in problem complexity. Other issues that arise with \textit{end-to-end} deep learning models is the lack of interpretability, which hinders the inclusion of robustness in these approaches. Also, as these approaches often rely on sensor data for inference, they may demand for active sensing and incur computational overheads during operation 
%\dk{---------------------After this point it looks ok}
\par To obtain a realistic approach that simultaneously does not rely on CSI model assumptions and/or network/IRS configurations, does not involve function approximators, and does not introduce additional overheads (e.g., due to cascaded channel estimation and/or continual IRS control),  a new oracle-based, data-driven and model-free zeroth-order method was recently proposed by the authors \cite{zosga_conf,IEEE_tran_Hashmietal}, termed ZoSGA in \cite{zosga_conf,IEEE_tran_Hashmietal},  which is able to deal directly with general nonconvex two-stage stochastic programs, without the utilization of convex approximations. It was proven, under very reasonable and realistic assumptions, that ZoSGA is (ergodically) convergent in the vicinity of some critical point of the (nonconvex) two-stage stochastic program at hand \cite{IEEE_tran_Hashmietal}. At the same time, it was numerically demonstrated, over a wide range of experiments in the context of sumrate maximization, that ZoSGA substantially outperforms state-of-the-art SSCA-based alternatives, despite the latter utilizing full knowledge of an exact channel model, its associated statistics, and the particular network/IRS configuration (cf. \cite{zosga_conf,IEEE_tran_Hashmietal}). Overall, by tackling the original nonconvex problem, it was demonstrated that oracle-based zeroth-order methods can jointly learn short-term beamformers and optimized long-term IRS parameters that yield substantially higher QoS at much lower computational/resource costs, while being theoretically supported under grounded assumptions.
\begin{figure}
  \centering
  \vspace{-3bp}
  \centerline{\includegraphics[width=3.3in]{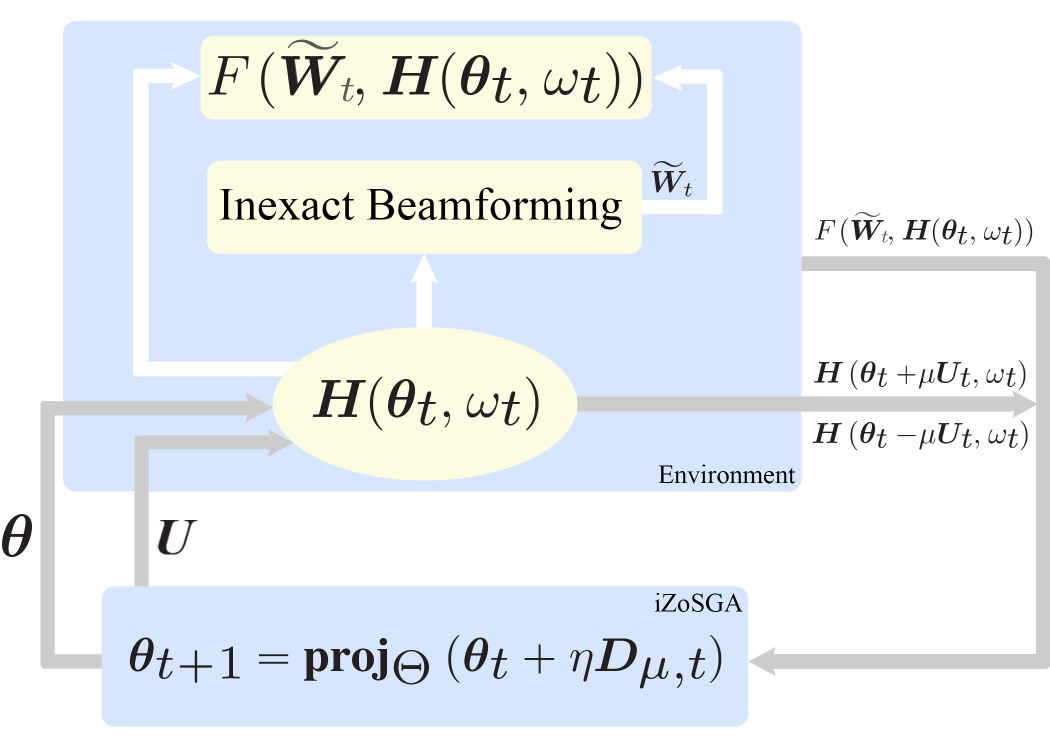}}
\vspace{-6bp}
\caption{Flowchart of the proposed algorithm, termed iZoSGA, which is the inexact extension of the so-called ZoSGA algorithm developed in \cite{IEEE_tran_Hashmietal}.}
\label{fig:graphic_izosga}
\vspace{-14pt}
\end{figure}
%\par The study in \cite{IEEE_tran_Hashmietal}, which to our knowledge was the first of its kind, contained several interesting technical observations concerning the inherent structure of the nonconvex two-stage stochastic problems that arise in the context of IRS-assisted utility maximization for optimal beamforming over wireless communication networks, and the associated convergence analysis was performed under very general assumptions. 

However, the study in \cite{IEEE_tran_Hashmietal}, which to our knowledge was the first of its kind, assumed access to a \emph{perfect oracle} (i.e., a ``black-box" returning a globally optimal solution to a given optimization problem) for short-term QoS maximization over feasible beamformers. 
%(in the adopted framework of two-stage stochastic programming, this corresponds to the deterministic second-stage problem; see Section \ref{subsec: problem formulation}).
While this assumption is significantly more general than typical alternative assumptions utilized in the relevant literature of bilevel (or two-stage) optimization (namely, the single-valuedness assumption of the second-stage optimal value function; e.g., see \cite{DanskinMinMax,Kwon_etal_ICML23}), it still creates a barrier between the theory and practice underlying the oracle-based zeroth-order algorithm proposed in \cite{IEEE_tran_Hashmietal}. Indeed, in the numerical experiments performed in \cite{IEEE_tran_Hashmietal}, the corresponding short-term precoding problems were approximately ``solved" by utilizing the well-known \emph{Weighted Minimum Mean Square Error} (WMMSE) algorithm \cite{wmmseShi2011} for a \textit{fixed} number of iterations, forgoing the exact oracle regime. Despite this discrepancy, the numerical results were extremely promising and demonstrated that ZoSGA in fact sets a new state of the art in the context of IRS-assisted optimal beamforming.

\par \textit{In this work}, we close this theory-practice gap, by introducing and developing an extension of the algorithm developed in \cite{IEEE_tran_Hashmietal}. This new algorithm, termed \textit{iZoSGA}, utilizes general \textit{inexact oracles}, which in fact might provide arbitrarily sub-optimal solutions to the underlying short-term beamforming problems; see also Fig. \ref{fig:graphic_izosga}. We analyze iZoSGA under mild assumptions, and establish its (ergodic) convergence rate in the vicinity of a stationary point of the original two-stage stochastic program, where the proximity to such a stationary point is now also directly controlled by the corresponding error of the inexact oracle (and where this error is measured in a certain mathematically precise sense; see Sections \ref{sec: conv anal}, \ref{sec: compatibility of assumptions}).
\par Our technical analysis reveals the oracle error propagation dynamics as the proposed algorithm iterates, initially without imposing any assumptions on the oracle itself, or its properties. Then, by specializing to a wide class of two-stage optimization problems that naturally appear in the context of short/long-term (passive) IRS-assisted optimal beamforming, we show that, in fact, the error of the inexact oracle can be controlled (to a reasonable extent). Consequently, we obtain a technical framework which explains the success and efficiency demonstrated in the numerical experiments presented in \cite{IEEE_tran_Hashmietal} in a reasonable and satisfying manner. 
%also gaining valuable insights concerning the propagation of error in the presence of inexact oracles returning approximate solutions of the associated second-stage (short-term beamforming) problems. 
We also gain valuable insghts on the ways the analysis presented herein can be used to inform algorithmic design \textit{after} taking oracle errors into account, which has the potential of enabling the design of specialized schemes of actual practical value. 

To further illuminate the findings above, we perform extensive numerical experiments using inexact oracles of varying quality. Our experiments empirically demonstrate the tolerance of iZoSGA relative to oracle inexactness, as well as showcase different algorithmic strategies that can be utilized to improve practical performance by adjusting the inexact oracle to the needs of the solver and of the associated optimization problem at hand. We illustrate that the algorithm can effectively handle various (even varying) levels of short-term inexactness during long-term optimization, in turn enabling tackling particularly large-scale instances of joint short/long-term IRS-assisted network utility maximization problems (in particular, with more than 38,000 links in our simulated environment). Finally, we verify that iZoSGA is truly agnostic to cascaded CSI models, network topologies and IRS configurations by considering a physical IRS model \cite{phy_em_model:costa2021electromagnetic}, successfully tuned \textit{on-the-fly}.
\par The paper is organized as follows. We present our problem formulation as well as  core technical results, as reported in \cite[Sections II, III]{IEEE_tran_Hashmietal}, in Section \ref{sec: Preliminaries}. Subsequently, in Section \ref{sec: conv anal}, we introduce iZoSGA and study its convergence properties at the presence of a general (arbitrary) inexact oracle. Then, 
%by focusing on IRS-aided optimal beamforming over wireless communication networks, 
in Section \ref{sec: compatibility of assumptions} we discuss the plausibility of controlling the oracle error, and study the effects of such an inexact oracle on the convergence rate of the proposed algorithm (as presented in Section \ref{sec: conv anal}). In Section \ref{sec: numerical results}, we present our extensive numerical results on a variety of large-scale problem settings, showcasing the robust behavior of the algorithm relative to the underlying oracle errors. Section \ref{sec: Conclusions} finally concludes the article.

\textit{Notation:} Hereafter, we use $\|\cdot\|$ to denote the standard Euclidean norm, defined as $\|\bm{x}\| \coloneqq \sqrt{\langle \bm{x}, \bm{x} \rangle}$ for $\bm{x} \in \mathbb{F}^n$, where $\mathbb{F}$ is a field (assuming that $\mathbb{F} = \mathbb{R}$ or $\mathbb{F} = \mathbb{C}$). In case of a complex vector we use the Hermitian inner product. In case of a matrix we consider the corresponding induced norm. Given a closed set $\mathcal{K}$, we define $\text{dist}(\bm{W},\mathcal{K}) \triangleq \inf_{\bm{Y} \in \mathcal{K}}\|\bm{Y} - \bm{W}\|$. We assume a complete base probability space $(\Xi,\mathscr{F},P)$, and use ``a.e." to denote ``$P$-almost every(where)". For $p \in [1,\infty)$, $\mathcal{Z}_p \equiv \mathcal{L}_p(\Xi,\mathscr{F},P;\mathbb{R})$ denotes the space of $\mathscr{F}$-measurable functions $\phi: \Xi \rightarrow \mathbb{R}$, such that $\int_{\Xi} \lvert \phi \rvert^p dP < \infty$. Given $f \colon \mathbb{R}^n \rightarrow \mathbb{R}$ and $\rho > 0$, we say that $f$ is $\rho$-weakly convex ($\rho$-weakly concave) if $f(\cdot) + \frac{\rho}{2}\|\cdot\|^2$ ($-f(\cdot) + \frac{\rho}{2}\|\cdot\|^2$) is convex.

\section{Problem Formulation and Preliminary Results} \label{sec: Preliminaries}
\subsection{Problem Formulation} \label{subsec: problem formulation}
\label{subsec: two-stage program}
Following the prior developments in \cite{zosga_conf,IEEE_tran_Hashmietal}, in this work we are interested in two-stage stochastic problems of the form
\begin{equation} \label{eqn: two-stage problem} \tag{2SP}
\begin{split}
    \boxed{
    \max_{\boldsymbol{\theta} \in \Theta}  \mathbb{E}\left\{\max_{\boldsymbol{W} \in \mathcal{W}}{F}\left(\bm{W},\boldsymbol{H}(\boldsymbol{\theta},\omega)\right) \right\},}
    \end{split}
\end{equation}
where $\mathcal{W}$ is a compact set of feasible dynamic (i.e., short-term) beamformers $\boldsymbol{W}$, and $\Theta\subset \mathbb{R}^S$ is a convex and compact set of feasible real-valued parameters that control the complex-valued IRS phase-shift elements (e.g., amplitudes and phases) present in the system. We let $F \colon \mathbb{C}^{M_U}\times \mathbb{C}^{M_U} \rightarrow \mathbb{R}$ be a known utility function depending on the beamforming vectors $\boldsymbol{W}$ and on the observable \textit{effective} or \textit{compound} random channel field 
 $\boldsymbol{H} \colon \mathbb{R}^S \times \Omega \rightarrow \mathbb{C}^{M_U}$, which in turn is a function of the IRS parameters $\boldsymbol{\theta}$, as well as a hidden random ``state of nature'' $\omega \colon \Xi \rightarrow \Omega$. The random vector $\omega$ represents anything that is unknown in regard to the channel $\boldsymbol{H}$, such as propagation or (compound) interference patterns, internal channel states, or in general the generation mechanism of underlying intermediate communication links. We also assume that the observable effective channels $\bm{H}(\cdot,\omega)$ have unknown dynamics, and that we are only allowed to evaluate them at specific IRS parameter values $\bm{\theta}\in \Theta$ via conventional channel estimation.
\par In what follows, we provide certain regularity assumptions on \eqref{eqn: two-stage problem} and subsequently present preliminary core technical results, as developed in \cite[Section II]{IEEE_tran_Hashmietal} (see Lemma \ref{lemma: gradient of compositional} and Theorem \ref{thm: gradient of f} below), enabling the development of the proposed optimization scheme, to be introduced shortly. 
%We note that the remaining of this section will be comprised of a collection of key results reported in \cite{IEEE_tran_Hashmietal}, which will form the basis for the novel developments in this work. 
For brevity of exposition, the associated proofs are omitted; the interested reader is referred to \cite{IEEE_tran_Hashmietal}.
\subsection{Assumptions}
\par \textit{Second-stage problem:} For any realization $\omega \in \Omega$, and any $\bm{\theta} \in \Theta$, the second-stage problem reads
\begin{equation} \label{eqn: second-stage problem} \tag{SSP}
\max_{\bm{W} \in \mathcal{W}} \,\, \left\{G(\bm{W},\bm{\theta},\omega)   \triangleq F\left(\bm{W},\bm{H}(\bm{\theta},\omega)\right)\right\}. 
\end{equation}
\noindent The second-stage problem \eqref{eqn: second-stage problem} is deterministic and is solved after the ---otherwise hidden--- state of nature $\omega$ is realized (i.e., for each instance of $\omega$).
\par \textit{First-stage problem:} The first-stage problem, which is an equivalent expression for \eqref{eqn: two-stage problem}, can be written as
\begin{equation} \label{eqn: first-stage problem} \tag{FSP}
\max_{\bm{\theta} \in \Theta} \,\, \{ f(\bm{\theta}) \triangleq \mathbb{E}\left\{F\left(\bm{W}^*(\bm{\theta},\omega),\bm{H}(\bm{\theta},\omega)\right) \right\}\} 
\end{equation}
\noindent where $f$ is defined with respect to any arbitrary $\bm{W}^*(\bm{\theta},\omega) \in \arg\max_{\bm{W} \in \mathcal{W}} F\left(\bm{W},\bm{H}(\bm{\theta},\omega)\right)$. In Assumption \ref{assumption: two-stage problem} that follows, we state certain regularity conditions enforced on \eqref{eqn: two-stage problem}. 
\begin{assumption} \label{assumption: two-stage problem}
The following conditions are in effect:
\begin{enumerate}[align=parleft,labelsep=0.75cm]
    \item[\textnormal{\textbf{(A1)}}] The function $F \colon \mathbb{C}^{M_U} \times \mathbb{C}^{M_U} \rightarrow \mathbb{R}$ is twice continuously (real) differentiable;
    \item[\textnormal{\textbf{(A2)}}] The sets $\Theta$ and $\mathcal{W}$ are compact, and $\Theta$ is also convex;
    \item[\textnormal{\textbf{(A3)}}] The function $\bm{H}(\cdot,\omega)$ is $B_H$-uniformly bounded on $\Theta$ and twice continuously differentiable on $\mathbb{R}^S$, for a.e. $\omega \in \Omega$. Moreover, there exist numbers $L_{H,0}$, $L_{H,1}$, such that $\bm{H}(\cdot,\omega)$ is $L_{H,0}$-Lipschitz continuous with $L_{H,1}$-Lipschitz gradients on $\Theta$ for a.e. $\omega \in \Omega$;
    \item[\textnormal{\textbf{(A4)}}] There exists a positive function $\widetilde{\rho}(\cdot) \in \mathcal{Z}_1$, such that the function $\max_{\bm{W} \in \mathcal{W}} F(\bm{W},\bm{H}(\cdot,\omega))$ is $\widetilde{\rho}(\omega)$-weakly concave on $\Theta$;
    \item[\textnormal{\textbf{(A5)}}]Independent and identically distributed (i.i.d.) samples from the law of $\omega$ are available (but hidden).
\end{enumerate}
\end{assumption}
\noindent The conditions in Assumption \ref{assumption: two-stage problem} are mild and met in a wide range of passive IRS-aided optimal beamforming optimization models appearing in the literature, and beyond. We also refer the reader to \cite[Remark 1, Section IV and Appendix B]{IEEE_tran_Hashmietal} for a detailed discussion showcasing the generality and applicability of these conditions.

\vspace{-4pt}
\subsection{Technical Results} \label{subsec: technical results}
\par Since ${F}(\bm{W},\cdot)$ takes a complex input, we utilize \emph{Wirtinger calculus} (see \cite{arXiv:Kreutz-Delgado}) to derive its gradient. The full representation of the \textit{compositional} gradient of ${F}(\bm{W},\bm{H}(\bm{\theta},\omega))$ realtive to $\boldsymbol{\theta}$ follows.
\begin{lemma} \label{lemma: gradient of compositional}
For every $\bm{\theta} \in \Theta$, $\bm{W} \in \mathcal{W}$ and a.e. $\omega\in \Omega$, the gradient of ${F}\left(\bm{W},\bm{H}(\bm{\theta},\omega)\right)$ with respect to $\bm{\theta}$ reads
\begin{align}
& \nabla_{\bm{\theta}} F\left(\bm{W},\bm{H}(\bm{\theta},\omega)\right) \label{eqn: Wirtinger gradient of the compositional objective} 
\\ &\quad =\ 2 \nabla_{\bm{\theta}} \Re\left(\bm{H}(\bm{\theta},\omega)\right)\Re\left(\frac{\partial^{\circ}}{\partial \bm{z}} F(\bm{W},\bm{z}) \bigg\vert_{\bm{z} = \bm{H}(\bm{\theta},\omega)} \right)^\top \nonumber
\\  & \qquad + 2 \nabla_{\bm{\theta}} \Im\left(\bm{H}(\bm{\theta},\omega)\right)\Re\left(j\frac{\partial^{\circ}}{\partial \bm{z}} F(\bm{W},\bm{z}) \bigg\vert_{\bm{z} = \bm{H}(\bm{\theta},\omega)} \right)^\top, \nonumber
\end{align}

\noindent where $\frac{\partial^{\circ}}{\partial \bm{z}}(\cdot)$ is the Wirtinger cogradient operator. Moreover, there exists a constant $B_F > 0$ such that
\[\left\|\frac{\partial^{\circ}}{\partial \bm{z}} F(\bm{W},\bm{z})\big\vert_{\bm{z} = \bm{H}(\bm{\theta},\omega)}\right\| \leq B_F.\]
\end{lemma}
\begin{proof}
   See \cite[Appendix A and Lemma 2]{IEEE_tran_Hashmietal}.
\end{proof}

\par Under Assumption \ref{assumption: two-stage problem}, the function $f$, as defined in \eqref{eqn: first-stage problem}, is, in fact, continuously differentiable and weakly concave on $\Theta$. This is established in the following theorem.
\begin{theorem} \label{thm: gradient of f}
     Let Assumption \textnormal{\ref{assumption: two-stage problem}} hold. Then, for any $\bm{\theta} \in \Theta$, the function $f$ is well-defined and differentiable, with
     \begin{equation} \label{eqn: gradient of f}
     \begin{split}
        \nabla_{\bm{\theta}} f(\bm{\theta}) =&\ \mathbb{E}\left\{\nabla_{\bm{\theta}} F\left(\bm{W},\bm{H}(\bm{\theta},\omega)\right)\big\vert_{\bm{W} = \bm{W}^*(\bm{\theta},\omega)}\right\},
        \end{split}
     \end{equation}
     \noindent for any $\bm{W}^*(\bm{\theta},\omega) \in \arg\max_{\bm{W} \in \mathcal{W}} F\left(\bm{W},\bm{H}(\bm{\theta},\omega)\right)$. Moreover, for a.e. $\omega \in \Omega$ there exists a constant $\widehat{\rho}(\omega) > 0$ such that the mapping $\theta \mapsto \max_{\bm{W} \in \mathcal{W}} F(\bm{W},\bm{H}(\bm{\theta},\omega))$ is $\widehat{\rho}(\omega)$-weakly convex on $\Theta$. If $\widehat{\rho}(\cdot) \in \mathcal{Z}_1$, $f$ is $\rho$-weakly concave on $\Theta$, where ${\rho} \triangleq \max\left\{\mathbb{E}\{\widetilde{\rho}(\omega)\},\mathbb{E}\{\widehat{\rho}(\omega)\}\right\}$.
\end{theorem}
\begin{proof}
\par Follows from \cite[Lemma 3 and Theorem 4]{IEEE_tran_Hashmietal}.
\end{proof}
\begin{remark}
    \par Hereafter, a mild condition on integrability of the weak convexity random parameter $\widehat{\rho}(\omega)$ is implicitly assumed to hold, i.e., $\widehat{\rho}(\cdot) \in \mathcal{Z}_1$.
\end{remark}

\subsection{Compound Channel Zeroth-order Gradient Approximation}

\par The channel field $\bm{H}(\cdot,\omega)$ is assumed to have unknown dynamics. Therefore, the proposed method will rely on gradient estimates arising from a two-point stochastic evaluation of $\bm{H}(\cdot,\omega)$ (e.g., see \cite{IEEE_tran_Hashmietal,SIAMOPT:KalogeriasPowellZerothOrder,CompMath:Nesterov_etal,arXiv:Pougk-Kal}). From \textbf{(A3)}, we have
\[\nabla_{\bm{\theta}} \bm{H}(\bm{\theta},\omega) = \nabla_{\bm{\theta}} \Re\left(\bm{H}(\bm{\theta},\omega)\right) + j \nabla_{\bm{\theta}} \Im\left(\bm{H}(\bm{\theta},\omega) \right).\]
\noindent We approximate the gradient using only function evaluations of $\bm{H}(\cdot,\omega)$. We let $\bm{U} \sim \mathcal{N}\left(\bm{0}, \bm{I}\right)$ be a normal random vector, where $\bm{I}$ is the identity matrix of size $S$. Given a smoothing parameter $\mu > 0$, we consider the following gradient estimate
\begin{align}\label{eqn: ZO gradient approximation}
& \nabla_{\bm{\theta}} \bm{H}_{\mu}(\bm{\theta},\omega) \\
&\ \  \triangleq \frac{1}{2\mu}\mathbb{E}\left\{\left(\bm{H}\left(\bm{\theta} \hspace{-1pt}+\hspace{-1pt} \mu \bm{U},\omega\right) \hspace{-1pt}-\hspace{-1pt} \bm{H}\left(\bm{\theta}-\mu \bm{U},\omega\right)\right)\bm{U}^\top \big\vert \ \omega \right\}^\top. \nonumber
\end{align}
\noindent Let us define the quantities $\nabla_{\bm{\theta}} \bm{H}^R_{\mu}(\bm{\theta},\omega) \triangleq  \Re\left(\nabla_{\bm{\theta}} \bm{H}_{\mu}(\bm{\theta},\omega)\right)$ and $\nabla_{\bm{\theta}} \bm{H}^I_{\mu}(\bm{\theta},\omega) \triangleq  \Im\left(\nabla_{\bm{\theta}} \bm{H}_{\mu}(\bm{\theta},\omega)\right)$. From \textbf{(A3)}, \eqref{eqn: ZO gradient approximation} is well-defined and valid for points $\bm{\theta}$ lying on the boundary of $\Theta$. 
\par The \textit{zeroth-order gradient approximation} of $\boldsymbol{H}(\cdot, \omega)$ constructed in \eqref{eqn: ZO gradient approximation}
has multiple operational benefits. Firstly, it enables us to bypass any modelling assumptions about the communication channel, which typically incur modelling errors, especially in large-scale systems. Additionally, \eqref{eqn: ZO gradient approximation} enables (approximate) gradient evaluation using \textit{any} real-valued isomorphism on the complex plane, e.g.,  polar coordinates. In other words, we can readily assume that the vector $\bm{\theta}$ represents any real-valued parameters determining the complex-valued phase-shift elements of the IRSs (e.g., amplitutes and phases). In turn, this allows us to completely bypass the typical nonconvex unit-modulus constraints that arise when optimizing over complex IRS phase-shifts.
\par Using \eqref{eqn: ZO gradient approximation}, we may lastly construct a zeroth-order gradient approximation of \eqref{eqn: Wirtinger gradient of the compositional objective} reading
\begin{align} \label{eqn: ZO gradient of the compositional function}
%\begin{split}
  & \nabla_{\bm{\theta}}^{\mu} F\left(\bm{W},\bm{H}(\bm{\theta},\omega)\right) \\
  &\quad \triangleq \ 2 \nabla_{\bm{\theta}} \bm{H}_{\mu}^R(\bm{\theta},\omega)\left(\Re\left(\frac{\partial^{\circ}}{\partial \bm{z}} F(\bm{W},\bm{z}) \big\vert_{\bm{z} = \bm{H}(\bm{\theta},\omega)} \right)\right)^\top \nonumber \\ & \qquad + 2\nabla_{\bm{\theta}}  \bm{H}_{\mu}^I(\bm{\theta},\omega)\left(\Re\left(j\frac{\partial^{\circ}}{\partial \bm{z}} F(\bm{W},\bm{z}) \big\vert_{\bm{z} = \bm{H}(\bm{\theta},\omega)} \right)\right)^\top, \nonumber
   %\end{split}
\end{align}
\noindent where $ \nabla_{\bm{\theta}}\bm{H}_{\mu}^R(\bm{\theta},\omega)$ and $ \nabla_{\bm{\theta}}\bm{H}_{\mu}^I(\bm{\theta},\omega)$ are defined after \eqref{eqn: ZO gradient approximation}. 
\section{Algorithm and Convergence Analysis} \label{sec: conv anal}

\par In this section, we develop our proposed zeroth-order projected stochastic gradient ascend method for tackling \eqref{eqn: first-stage problem}. To this end, we define the notion of an \emph{inexact} oracle able to find points that are ``close" (in a certain sense) to an arbitrary optimal solution of the deterministic problem \eqref{eqn: second-stage problem}. 
\begin{definition} \label{definition: oracle assumption}
Let Assumption \textnormal{\ref{assumption: two-stage problem}} hold for problem \textnormal{\eqref{eqn: two-stage problem}}. Given any $\bm{\theta} \in \Theta$ and for a.e. $\omega \in \Omega$, an \emph{inexact oracle} exists that returns an approximate solution $\widetilde{\bm{W}} \in \mathcal{W}$ to \textnormal{\eqref{eqn: second-stage problem}}, such that
\begin{equation} \label{eqn: solution distance error} \mathrm{dist}\Big(\widetilde{\bm{W}},\arg\max_{\bm{W} \in \mathcal{W}} F(\bm{W},\bm{H}(\bm{\theta},\omega))\Big) \leq \varepsilon(\bm{\theta},\omega),
\end{equation}
\noindent with $\varepsilon(\cdot,\cdot)$ some measurable random error function.
\end{definition}
\par Let us observe that Definition \ref{definition: oracle assumption} does not specify the magnitude of the random error function, but only its measurability. Indeed, Definition \ref{definition: oracle assumption} introduces the notion of an inexact oracle with output $\widetilde{\bm{W}}$, in the sense that there exists $\bm{W}^* \in \arg\max_{\bm{W} \in \mathcal{W}} F(\bm{W},\bm{H}(\bm{\theta},\omega))$, such that $\|\widetilde{\bm{W}} - \bm{W}^*\| \leq \varepsilon(\bm{\theta},\omega)$, without specifying anything else about $\varepsilon(\cdot)$. This is intentional, since it allows for great flexibility when looking for approximate solutions to the second-stage problem \eqref{eqn: second-stage problem}. For example, one might attempt to find an inexact solution to \eqref{eqn: second-stage problem} by utilizing some iterative algorithm (such as the well-known WMMSE algorithm \cite{wmmseShi2011}, in case of sumrate maximization); then, $\varepsilon$ would be completely determined by the properties of the sequences generated by that iterative scheme.

\par Next, we will develop and analyze our proposed algorithm by exploiting such a general inexact oracle, and showcase how the (averaged-over-the-iterates expected) random oracle errors propagate in long-term IRS optimization, when using the oracle's produced inexact inner (i.e., short-term) solutions $\widetilde{\bm{W}}$. Subsequently, in Section \ref{sec: compatibility of assumptions}, we will show that for a wide class of problems of interest (completely compatible with our core conditions given in Assumption \ref{assumption: two-stage problem}, and with the application of interest), this error function can be expected to have a controlled magnitude (despite the potential non-concavity of the second-stage problem \eqref{eqn: second-stage problem}). 

\vspace{-1pt}\subsection{A Zeroth-order Projected Stochastic Gradient Ascent}
\par We now present the proposed (model-free) zeroth-order projected stochastic gradient ascent that uses inexact oracles for the second stage problem. From \textbf{(A5)}, we have available i.i.d. samples following the law of $\omega$. Thus, we may sample \eqref{eqn: ZO gradient of the compositional function} at every $(\bm{\theta},\bm{W}) \in \Theta\times\mathcal{W}$ and for a.e. $\omega \in \Omega$, to obtain the zeroth-order stochastic gradient approximation
\begin{align} \label{eqn: ZO gradient sample of the compositional function}
 & \bm{D}_{\mu}(\bm{\theta},\omega,\bm{U};\bm{W})  \triangleq \bm{\Delta}_{\mu}^R \left(\Re\left(\frac{\partial^{\circ}}{\partial \bm{z}} F(\bm{W},\bm{z}) \bigg\vert_{\bm{z} = \bm{H}(\bm{\theta},\omega)} \right)\right)^\top \nonumber\\ & \qquad   + \bm{\Delta}_{\mu}^I \left(\Re\left(j\frac{\partial^{\circ}}{\partial \bm{z}} F(\bm{W},\bm{z}) \bigg\vert_{\bm{z} = \bm{H}(\bm{\theta},\omega)} \right)\right)^\top,
\end{align}
\begin{align}\nonumber
    &\hspace{-8pt}\left(\bm{\Delta}_{\mu}^R, \bm{\Delta}_{\mu}^I \right)  
    \\ 
    &\triangleq\frac{1}{\mu}\Big[\left(\Re\left(\bm{\Delta}_{\mu}(\bm{\theta},\omega,\bm{U})\right) \bm{U}^\top \right)^\top \, \left(\Im\left(\bm{\Delta}_{\mu}(\bm{\theta},\omega,\bm{U})\right) \bm{U}^\top \right)^\top\Big], \nonumber
\end{align}
\noindent where $\mu > 0$ is the smoothing parameter, $\bm{U} \sim \mathcal{N}(\bm{0},\bm{I})$ and $\bm{\Delta}_{\mu}(\bm{\theta},\omega,\bm{U}) \triangleq \bm{H}\left(\bm{\theta} + \mu \bm{U},\omega\right) - \bm{H}\left(\bm{\theta}-\mu \bm{U},\omega\right)$. This sample gradient is obtained by probing the wireless network twice on perturbations of IRS parameters $\bm{\theta} + \mu \bm{U}$ and $\bm{\theta} - \mu \bm{U}$. This induces an overhead of two channel estimations, implicitly assumed to be within the coherence time of $\bm{H}$. Such an assumption is minimal and standard in the related literature on resource allocation in wireless systems; see, e.g., \cite{IEEE_tran_Hashmietal,2019eisen_pfo}.

\par The proposed algorithm evaluates \eqref{eqn: ZO gradient sample of the compositional function} at $\boldsymbol{W}=\widetilde{\boldsymbol{W}}_t$, and uses the resulting expression as an ascend direction in a projected quasi-gradient step. The scheme is summarized in Algorithm \ref{Algorithm: iZoSGA}; see also Fig. \ref{fig:graphic_izosga} for a concept diagram. We should also note that evaluating \eqref{eqn: ZO gradient sample of the compositional function} (at $\boldsymbol{W}=\widetilde{\boldsymbol{W}}_t$) does not correspond to a zeroth-order stochastic gradient representation of any surrogate function. Instead, this is a crude approximation of the gradient of the second-stage max-function appearing in \eqref{eqn: second-stage problem}, arising from the utilization of Gaussian smoothing, at the presence of an inexact oracle.

\par Overall, as also explained in \cite{IEEE_tran_Hashmietal}, the proposed algorithm requires exactly three effective channel estimations, per operational iteration: one to communicate (on which the short-term precoders $\widetilde{\boldsymbol{W}}$ are calculated), and two more pertaining to the system probes required for the sample gradient approximations. Further, we emphasize that channel estimation can be performed using any conventional scheme, since it is independent of the presence of IRSs in the system.
\renewcommand{\thealgorithm}{iZoSGA}
\begin{algorithm}[!t]
\caption{}
    \label{Algorithm: iZoSGA}
\begin{algorithmic}
\State \textbf{Input:}  $\bm{\theta}_0 \in \Theta$, $\eta > 0$, $\mu > 0$, $\varepsilon(\cdot) > 0$, and $T > 0$.
\For {($t = 0,1,2,\ldots, T$)}
\State Sample (i.i.d.) $\omega_t \in \Omega$, $\bm{U}_t \sim \mathcal{N}\left(\bm{0},\bm{I}\right)$.
\State Find $\widetilde{\bm{W}}_t \in \mathcal{W}$ such that
\[ \text{dist}\left(\widetilde{\bm{W}}_t,\arg\max_{\bm{W} \in \mathcal{W}} F(\bm{W},\bm{H}(\bm{\theta}_t,\omega_t))\right) \leq \varepsilon(\bm{\theta}_t,\omega_t).\]
\State Set $\bm{D}_{\mu,t} \equiv \bm{D}_{\mu}\left(\bm{\theta}_t,\omega_t,\bm{U}_t;\widetilde{\bm{W}}_t\right)$ as in \eqref{eqn: ZO gradient sample of the compositional function}.
\State  $\bm{\theta}_{t+1} = \textbf{proj}_{\Theta}\left(\bm{\theta}_t + \eta \bm{D}_{\mu,t}\right). $
\EndFor
\State Sample $t^* \in \{0,\ldots,T\}$ according to $\mathbb{P}(t^* = t) = \frac{1}{T+1}$.
\State \Return $\bm{\theta}_{t^*}$.
\end{algorithmic}
\end{algorithm}

\subsection{Convergence Analysis}
\par We proceed by establishing the convergence of \ref{Algorithm: iZoSGA}. In order to perform our analysis, we introduce some additional notation for convenience. Specifically, given $\bm{\theta} \in \mathbb{R}^S$ and $\bm{W}(\omega) \in \mathcal{W}$, let
\[ f_{\bm{W}}(\bm{\theta}) \triangleq \mathbb{E}\left\{ F(\bm{W},\bm{H}(\bm{\theta},\omega))\right\}.\]
\begin{lemma} \label{lemma: auxiliary function properties}
    Let Assumption \textnormal{\ref{assumption: two-stage problem}} hold. There exists a constant $L_{f,1} > 0$ such that for any $\bm{\theta} \in \Theta$ and any two random vectors $\bm{W}_1(\omega), \bm{W}_2(\omega) \in \mathcal{W}$, it holds that
    \begin{equation} \label{eqn: cross Lipschitz smooth} \|\nabla_{\bm{\theta}} f_{\bm{W}_1}(\bm{\theta}) - \nabla_{\bm{\theta}} f_{\bm{W}_2}(\bm{\theta})\| \leq L_{f,1} \mathbb{E}\left\{\|\bm{W}_1 - \bm{W}_2\|\right\}.
    \end{equation}
\end{lemma}
\begin{proof}
\iffalse
\par Firstly, we observe that conditions \textbf{(A1)} and \textbf{(A3)} of Assumption \ref{assumption: two-stage problem} imply that $F(\cdot,\bm{H}(\cdot,\omega))$ is twice continuously (real) differentiable on the compact set $\mathcal{W} \times \Theta$. On the other hand, since $\Theta$ is comapct, and by utilizing \cite[Theorem 7.44]{SIAM:Shapiro_etal}, we observe that $f_{\bm{W}}$ is Lipschitz smooth on $\Theta$. Finally, from the definition of $\rho$ in Theorem \ref{thm: gradient of f} we immediately observe that $f_{\bm{W}}$ must be $\rho$-Lipschitz smooth for any $\bm{W} \in \mathcal{W}$.
\fi
\par We first note that $\nabla_{\bm{\theta}} f_{\bm{W}}(\bm{\theta}) = \mathbb{E}\left\{\nabla_{\bm{\theta}} F(\bm{W},\bm{H}(\bm{\theta},\omega)\right\}$ (from \cite[Theorem 7.44]{SIAM:Shapiro_etal}). Then, from Lemma \ref{lemma: gradient of compositional}, we obtain the expression for  $\nabla_{\bm{\theta}} F(\bm{W},\bm{H}(\bm{\theta},\omega))$, for any $\bm{\theta} \in \Theta$, and any $\bm{W}(\omega) \in \mathcal{W}$, where $\omega \in \Omega$ (see \eqref{eqn: Wirtinger gradient of the compositional objective}). By substituting this gradient expression in the left side of \eqref{eqn: cross Lipschitz smooth}, we observe that the statement follows from \textbf{(A1)} of Assumption \ref{assumption: two-stage problem}. Indeed, from \textbf{(A1)} we obtain that $F(\cdot,\cdot)$ is twice (real) continuously differentiable. Hence, the Wirtinger cogradient appearing in \eqref{eqn: Wirtinger gradient of the compositional objective} is Lipschitz smooth with respect to $\bm{W}$, since $\mathcal{W}$ is compact, and the definition of the Wirtinger cogradient depends on the real gradients of $F(\bm{W},\cdot)$ (see \cite[Section 4.2]{arXiv:Kreutz-Delgado}). The proof then follows by a simple application of Jensen's inequality.
\end{proof}
\begin{lemma} \label{lemma: bound on sample gradient variance}
Let Assumption  \textnormal{\ref{assumption: two-stage problem}} hold, and fix any $\bm{\theta} \in \Theta$ and any random $\bm{W}(\omega) \in \mathcal{W}$. For a.e. $\omega \in \Omega$, any $\mu \geq 0$, and any $\bm{U} \sim \mathcal{N}(\bm{0},\bm{I})$, the following hold:
\begin{equation} \label{eqn: bound on the expected gradient}
\mathbb{E}\left\{\left\| \bm{D}_{\mu}(\bm{\theta},\omega,\bm{U};\bm{W})\right\|^2\right\} \leq  4B_F^2 L_{H,0}^2(S^2 + 2S),
\end{equation}
\begin{equation*}
\scalebox{0.96}{$\begin{split}
\mathbb{E}\left\{\bm{D}_{\mu}(\bm{\theta},\omega,\bm{U};\bm{W})\right] \triangleq   \mathbb{E}\left\{\nabla_{\bm{\theta}}^{\mu} F\left(\bm{W},\bm{H}(\bm{\theta},\omega)\right) \right\} \triangleq \widehat{\nabla} f_{\bm{W}}(\bm{\theta}),
\end{split}$}
\end{equation*}
\noindent and 
\begin{equation*}
\begin{split}
 \left\| \widehat{\nabla}f_{\bm{W}}(\bm{\theta}) - \nabla f_{\bm{W}}(\bm{\theta}) \right\| \leq  2 \mu B_F L_{H,1} \sqrt{M S},
\end{split}
\end{equation*}
\noindent for $\bm{U}$ and $\omega$ being statistically independent.
\end{lemma}
\begin{proof}Trivial extension of \cite[Lemmata 5, 6]{IEEE_tran_Hashmietal}.
\end{proof}
\textit{Moreau Envelope:} We write the objective function of \eqref{eqn: first-stage problem} as $\phi(\bm{\theta}) \triangleq -f(\bm{\theta}) + \delta_{\Theta}(\bm{\theta})$, where $\delta_{\Theta}(\bm{\theta})$ denotes the indicator function for the convex compactum $\Theta$. Then, given a penalty parameter $\lambda > 0$, we define the \textit{proximity operator} as 
\[\textbf{prox}_{(1/\lambda) \phi}(\bm{u}) \triangleq \underset{\bm{\theta} \in \mathbb{R}^S}{\arg\min} \left\{\phi(\bm{\theta}) + \frac{\lambda}{2}\|\bm{u}-\bm{\theta}\|^2 \right\}, \]
\noindent and the corresponding Moreau envelope as
\[ \phi^{1/\lambda}(\bm{u}) \triangleq \min_{\bm{\theta} \in \mathbb{R}^S} \left\{\phi(\bm{\theta}) + \frac{\lambda}{2}\|\bm{u}-\bm{\theta}\|^2 \right\}.\]
\noindent Upon noting that $\phi$ is $\rho$-weakly convex, we observe that the Moreau envelope with parameter $\lambda > \rho$ is smooth even if $\phi(\cdot)$ is not, and the magnitude of its gradient can be used as a near-stationarity measure of the non-smooth problem of interest. If a point $\bm{\theta}$ is $\epsilon$-stationary for the Moreau envelope, then it is close to a near-stationary point of \eqref{eqn: two-stage problem}. This is a  standard approach of measuring progress in the context of weakly convex optimization (see \cite{SIAMOpt:Davis,arXiv:Pougk-Kal}, among others). 
\par Note that the definition of $f$ depends on the particular choice of solution to the second-stage problem \eqref{eqn: second-stage problem}. Nonetheless, the values of $f$ and its gradients are the same for any such choice, and thus we have the freedom of choosing the second-stage solution arbitrarily. This fact is used in the subsequent convergence analysis and is important to keep in mind. The following result establishes a convergence rate for \ref{Algorithm: iZoSGA}.

\begin{theorem} \label{thm: convergence analysis}
Assume that $\{\bm{\theta}_t\}_{t = 0}^T$, $T > 0$ is generated by \textnormal{\ref{Algorithm: iZoSGA}}, where $\bm{\theta}_{t^*}$ is the point that the algorithm returns. Let Assumption \textnormal{\ref{assumption: two-stage problem}} be in effect and, for each $t$, let $\varepsilon_t(\bm{\theta}_t,\omega_t) = \|\widetilde{\bm{W}}_t - \bm{W}_t^*\|$ be the oracle error at iteration $t$, where $\bm{W}_t^*$ is some solution to the second-stage problem given $\bm{\theta}_t$ and $\omega_t$. For any $\bar{\rho} > \rho$, it holds that
\begin{align}
      &  \mathbb{E}\left\{\left\| \nabla \phi^{1/\bar{\rho}}(\bm{\theta}_{t^*})\right\|^2\right\} \label{eqn: returning point Moreau gradient}\\
      &   \leq \hspace{-2pt}
      \frac{\bar{\rho}}{\bar{\rho}-\rho}\Bigg( \frac{\phi^{1/\bar{\rho}}(\bm{\theta}_0) - \min \phi(\bm{\theta}) +  \frac{1}{2}C_2 \bar{\rho}(T+1) \eta^2 }{(T+1)\eta} 
      \hspace{-1pt}+\hspace{-1pt}
      \frac{\bar{\rho}}{2}\widebar{C}_1 \Bigg), \nonumber
\end{align}
\noindent where 
\[\widebar{C}_1 \triangleq 2\Delta_{\Theta}\left(\frac{L_{f,1}}{T+1} \sum_{t=0}^{T}\mathbb{E}\{\varepsilon(\bm{\theta}_t,\omega_t)\} + 2\mu B_F L_{H,1}\sqrt{MS} \right),\]
\noindent and $C_2 = 4 B_F^2 L_{H,0}^2(S^2+2S)$. Moreover, if $\bar{\rho} = 2\rho$ and
\[ \eta = \sqrt{\frac{\Delta_f}{C_2 \rho  (T+1)}},\]
\noindent for some $\Delta_f \geq \phi^{1/(2\rho)}(\bm{\theta}_0) - \min  \phi(\bm{\theta})$, then it holds that
\begin{equation} \label{eqn: bound on Moreau envelope}
    \begin{split}
     \mathbb{E}\left\{\left\| \nabla \phi^{1/(2\rho)}(\bm{\theta}_{t^*})\right\|^2\right\} \leq &\ 8\left(\sqrt{\frac{\Delta_f \rho C_2}{4(T+1)}} + \frac{1}{2}\widebar{C}_1\rho \right).
    \end{split}
\end{equation}
\end{theorem}
\begin{proof}
    For any $t \geq 0$, we let  $\hat{\bm{\theta}}_{t} \triangleq \textbf{prox}_{\phi/\bar{\rho}}(\bm{\theta}_{t})$, and $\mathbb{E}_{[t]}\left\{\cdot\right\} \triangleq \mathbb{E}\left\{ \vert \omega_{t-1},\mathbf{U}_{t-1},\ldots,\omega_0,\mathbf{U}_0\right\}.$ Algorithm \ref{Algorithm: iZoSGA} yields
    \begin{equation*}
        \begin{split}
            &\mathbb{E}_{[t]} \left\{ \left\|\hat{\bm{\theta}}_{t} - \bm{\theta}_{t+1}\right\|^2\right\} =   \mathbb{E}_{[t]} \left\{ \left\|\hat{\bm{\theta}}_{t} - \textbf{Proj}_{\Theta}\left(\bm{\theta}_{t} + \eta \mathbf{D}_{\mu,t} \right)\right\|^2\right\}\\
            &\quad = \mathbb{E}_{[t]} \left\{ \left\|\textbf{Proj}_{\Theta}\left(\hat{\bm{\theta}}_{t}\right) - \textbf{Proj}_{\Theta}\left(\bm{\theta}_{t} + \eta \mathbf{D}_{\mu,t} \right)\right\|^2\right\}\\
            &\quad \leq \mathbb{E}_{[t]} \left\{ \left\|\hat{\bm{\theta}}_{t} - \bm{\theta}_{t} - \eta \mathbf{D}_{\mu,t} \right\|^2\right\}\\
            &\quad \leq \left\|\hat{\bm{\theta}}_{t} - \bm{\theta}_{t}\right\|^2  - \mathbb{E}_{[t]}\left\{ 2\eta\left\langle \hat{\bm{\theta}}_{t} - \bm{\theta}_{t},  \mathbf{D}_{\mu,t} \right\rangle  -\eta^2 \|\bm{D}_{\mu,t}\|^2\right\}\\
            &\quad \leq \left\|\hat{\bm{\theta}}_{t} - \bm{\theta}_{t}\right\|^2  -  2\eta\left\langle \hat{\bm{\theta}}_{t} - \bm{\theta}_{t},  \widehat{\nabla} f_{\widetilde{\bm{W}}_t}(\bm{\theta}_t) \right\rangle  \\ &\qquad +4\eta^2 B_F^2 L_{H,0}^2(S^2+2S), \\
        \end{split}
    \end{equation*}
    \noindent where in the first inequality we used the nonexansiveness of the projection operator, and in the last inequality we utilized Lemma \ref{lemma: bound on sample gradient variance}. Next, from the third part of Lemma \ref{lemma: bound on sample gradient variance} we obtain
     \begin{equation*}
        \begin{split}
            &\mathbb{E}_{[t]} \left\{ \left\|\hat{\bm{\theta}}_{t} - \bm{\theta}_{t+1}\right\|^2\right\}\\
            &\quad \leq \left\|\hat{\bm{\theta}}_{t} - \bm{\theta}_{t}\right\|^2  -  2\eta\left\langle \hat{\bm{\theta}}_{t} - \bm{\theta}_{t}, \nabla f(\bm{\theta}_t) \right\rangle  \\ &\qquad - 2\eta\left\langle \hat{\bm{\theta}}_{t} - \bm{\theta}_{t},   \widehat{\nabla} f_{\widetilde{\bm{W}}_t}(\bm{\theta}_t)-\nabla f_{\widetilde{\bm{W}}_t}(\bm{\theta}_t) \right\rangle \\ 
            &\qquad - 2\eta\left\langle \hat{\bm{\theta}}_{t} - \bm{\theta}_{t},  \nabla f_{\widetilde{\bm{W}}_t}(\bm{\theta}_t) -\nabla f(\bm{\theta}_t) \right\rangle\\
            &\qquad +4\eta^2 B_F^2 L_{H,0}^2(S^2+2S) \\
            &\quad \leq \left\|\hat{\bm{\theta}}_{t} - \bm{\theta}_{t}\right\|^2  -  2\eta\left\langle \hat{\bm{\theta}}_{t} - \bm{\theta}_{t}, \nabla f(\bm{\theta}_t) \right\rangle  \\
            &\qquad + 2\eta \Delta_{\Theta} L_{f,1} \mathbb{E}_{[t]}\left\{\left\|\widetilde{\bm{W}}_t - \bm{W}_t^*\right\|\right\}\\
            &\qquad + 4\eta \mu\Delta_{\Theta}B_F L_{H,1}\sqrt{MS}  +4\eta^2 B_F^2 L_{H,0}^2(S^2+2S), \\
        \end{split}
    \end{equation*}
    \noindent where we used \eqref{eqn: cross Lipschitz smooth}, and we let $\Delta_{\Theta}$ be the diamater of the compact set $\Theta$, and $\bm{W}_t^*$ be the solution of the second-stage problem (defined by $\bm{\theta}_t$ and $\omega_t$) attaining \eqref{eqn: solution distance error} (assuming that this very point is used to define the function $f$; note that we have the freedom of choosing any $\bm{W}^*$ from the solution set of the second-stage problem). Thus, for any $\bar{\rho} > \rho$, we obtain
\begin{equation*}
    \begin{split}
       & \mathbb{E}_{[t]}\left\{\phi^{1/\bar{\rho}}(\bm{\theta}_{t+1})\right\} \leq  \mathbb{E}_{[t]}\left\{f(\bar{\bm{\theta}}_{t}) +  \frac{\bar{\rho}}{2}\big\|\hat{\bm{\theta}}_t - \bm{\theta}_{t+1}\big\|^2\right\}\\
      &\quad  \leq \phi\big(\hat{\bm{\theta}}_t\big) + \frac{\bar{\rho}}{2}\left\|\hat{\bm{\theta}}_{t} - \bm{\theta}_{t}\right\|^2  -  \bar{\rho}\eta\left\langle \hat{\bm{\theta}}_{t} -\bm{\theta}_{t}, \nabla f(\bm{\theta}_t) \right\rangle \\
      &\qquad + \frac{\bar{\rho}}{2}\left(\eta  C_1 +\eta^2 C_2 \right)\\
     &\quad \leq  \phi^{1/\bar{\rho}}(\bm{\theta}_t) -  \bar{\rho}\eta\left\langle \hat{\bm{\theta}}_{t} -\bm{\theta}_{t}, \nabla f(\bm{\theta}_t) \right\rangle + \frac{\bar{\rho}}{2}\left(\eta  C_1 +\eta^2 C_2 \right)\\
       &\quad \leq  \phi^{1/\bar{\rho}}(\bm{\theta}_t) + \bar{\rho}\eta\left(f(\bm{\theta}_t)-f\big(\hat{\bm{\theta}}_t\big) + \frac{\rho}{2}\left\|\bm{\theta}_t - \hat{\bm{\theta}}_t\right\|^2\right) \\
       &\qquad +\frac{\bar{\rho}}{2}\left(C_1 \eta +  C_2\eta^2\right),
    \end{split}
\end{equation*}
\noindent where in the second inequality we have substituted $C_1  \triangleq 2\Delta_{\Theta}\left(L_{f,1} \mathbb{E}_{[t]}\{\varepsilon(\bm{\theta}_t,\omega_t)\} + 2\mu B_F L_{H,1}\sqrt{MS} \right)$ and $C_2 \triangleq 4 B_F^2 L_{H,0}^2(S^2+2S)$, in the third inequality we used the definition of the Moreau envelope along with the definition of $\hat{\bm{\theta}}_t$, and in the fourth inequality we used the weak convexity of $-f(\cdot)$. Next, we follow the developments in \cite[Section 3.1]{SIAMOpt:Davis}, by noting that the mapping $\bm{\theta} \mapsto -f(\bm{\theta}) + \frac{\bar{\rho}}{2}\|\bm{\theta}-\bm{\theta}_t\|^2$ is strongly convex with parameter $\bar{\rho}-\rho$, and is minimized at $\hat{\bm{\theta}}_t$. Hence,
\begin{equation*}
\begin{split}
& f(\hat{\bm{\theta}}_t) - f(\bm{\theta}_t)  - \frac{\rho}{2}\|\bm{\theta}_t - \hat{\bm{\theta}}_t\|^2 =  \left(- f(\bm{\theta}_t)  + \frac{\bar{\rho}}{2}\left\|\hat{\bm{\theta}}_t - \hat{\bm{\theta}}_t\right\|^2\right)   \\ &\qquad -\left(-f(\hat{\bm{\theta}}_t) + \frac{\bar{\rho}}{2}\left\|{\bm{\theta}}_t - \hat{\bm{\theta}}_t\right\|^2 \right)  + \frac{\bar{\rho}-\rho}{2}\left\|\bm{\theta}_t - \hat{\bm{\theta}}_t\right\|^2  \\
& \quad \geq (\bar{\rho}-\rho)\left\|\bm{\theta}_t - \hat{\bm{\theta}}_t\right\|^2 \equiv \frac{\bar{\rho}-\rho}{\bar{\rho}^2}\left\|\nabla \phi^{1/\bar{\rho}}(\bm{\theta}_t) \right\|^2, 
\end{split}
\end{equation*}
\noindent where the last equivalence follows from \cite[Lemma 2.2]{SIAMOpt:Davis}. Thus, by combining the last two inequalities, we obtain
\begin{equation*}
    \begin{split}
       \mathbb{E}_{[t]}\left\{\phi^{1/\bar{\rho}}(\bm{\theta}_{t+1})\right\} \leq &\   \phi^{1/\bar{\rho}}(\bm{\theta}_t) - \frac{\eta (\bar{\rho}-\rho)}{\bar{\rho}}\left\| \nabla \phi^{1/\bar{\rho}}(\bm{\theta}_t)\right\|^2 \\&\quad + \frac{\bar{\rho}}{2}\left(C_1 \eta +  C_2\eta^2\right).
    \end{split}
\end{equation*}
\noindent Taking expectations with respect to the filtration $\omega_0$, $\bm{U}_0$, $\ldots,\ \omega_{t-1}$, $\bm{U}_{t-1}$ and using the law of total expectation, yields
\begin{equation*} \label{eqn: convergence analysis Moreau envelope expected descent}
\begin{split}
    & \mathbb{E}\left\{ \phi^{1/\bar{\rho}}(\bm{\theta}_{t+1}) \right\} \leq   \mathbb{E}\left\{\phi^{1/\bar{\rho}}(\bm{\theta}_t)\right\} \\
    &\quad + {\bar{\rho}} \Delta_{\Theta} \eta \left(L_{f,1} \mathbb{E}\{\varepsilon(\bm{\theta}_t,\omega_t)\} + 2\mu B_F L_{H,1}\sqrt{MS} \right)  \\   % & \quad \leq   \mathbb{E}\left\{\phi^{1/\bar{\rho}}(\bm{\theta}_t)\right\}  -  \frac{\eta_t(\bar{\rho}-\rho)}{\bar{\rho}} \mathbb{E}\left\{\left\| \nabla \phi^{1/\bar{\rho}} (\bm{\theta}_t)\right\|^2\right\} 
   &\quad  + \frac{\bar{\rho}}{2}C_2\eta^2 -  \frac{\eta_t(\bar{\rho}-\rho)}{\bar{\rho}} \mathbb{E}\left\{\left\| \nabla \phi^{1/\bar{\rho}} (\bm{\theta}_t)\right\|^2\right\}.
     \end{split}
\end{equation*}
\noindent Subsequently, we can unfold the last inequality to obtain
\begin{equation*}
    \begin{split}
        &\mathbb{E}\left\{\phi^{1/\bar{\rho}}(\bm{\theta}_{T+1})\right\} \leq \phi^{1/\bar{\rho}}(\bm{\theta}_0) +C_2 \frac{\bar{\rho}}{2}(T+1) \eta^2 \\ 
        & \ + {\bar{\rho}} \Delta_{\Theta} \eta \left(L_{f,1} \sum_{t=0}^{T}\mathbb{E}\{\varepsilon(\bm{\theta}_t,\omega_t)\} + 2\mu B_F L_{H,1}\sqrt{MS} \right) \eta \\&\  - \frac{\bar{\rho}-\rho}{\bar{\rho}} \sum_{t=0}^T \eta_t\mathbb{E}\left\{\left\| \nabla \phi^{1/\bar{\rho}}(\bm{\theta}_t)\right\|^2\right\}.
    \end{split}
\end{equation*}
\noindent Then, we bound the left-hand side from below by $\phi(\bm{\theta}^*) \triangleq \min_{\bm{\theta} \in \Theta} f(\bm{\theta})$, and rearrange, to obtain
\begin{equation*}
    \begin{split}
   & \frac{1}{T+1} \sum_{t=0}^T  \mathbb{E}\left\{\left\| \nabla \phi^{1/\bar{\rho}}(\bm{\theta}_t)\right\|^2\right\} \\
   &\quad \leq  \frac{\bar{\rho}}{\bar{\rho}-\rho}\Bigg( \frac{\phi^{1/\bar{\rho}}(\bm{\theta}_0) - \phi(\bm{\theta}^*) + (1/2)C_2\bar{\rho} (T+1)\eta^2 }{(T+1)\eta}\Bigg) \\
   &\qquad   +\frac{\bar{\rho}^2}{2(\bar{\rho}-\rho)} \widebar{C}_1,
    \end{split}
\end{equation*}
\noindent where
\[ \widebar{C}_1 \triangleq 2\Delta_{\Theta}\left(\frac{L_{f,1}}{T+1} \sum_{t=0}^{T}\mathbb{E}\{\varepsilon(\bm{\theta}_t,\omega_t)\} + 2\mu B_F L_{H,1}\sqrt{MS} \right).\]
\noindent Since the left-hand side is exactly $\mathbb{E}\{\| \nabla \phi^{1/\bar{\rho}}(\bm{\theta}_{t^*})\|^2\}$, we deduce that \eqref{eqn: returning point Moreau gradient} holds. Finally, setting $\bar{\rho} = 2\rho$, letting ${\Delta_f} \geq \phi^{1/(2\rho)}(\bm{\theta}_0) - \min \phi(\bm{\theta})$, and choosing the step size as
\[ \eta = \sqrt{\frac{\Delta_f}{C_2(T+1)}},\]
\noindent we obtain \eqref{eqn: bound on Moreau envelope} which completes the proof.
\end{proof}
 \label{remark: complexity}
Let us define the averaged-over-the-iterates expected error
\[ \bar{\varepsilon} \triangleq \frac{1}{T+1} \sum_{t=0}^T \mathbb{E}\left\{\varepsilon(\bm{\theta}_t,\omega_t)\right\},\]
\noindent where $\{(\bm{\theta}_t,\omega_t)\}_{t = 0}^T$ are produced by Algorithm \textnormal{\ref{Algorithm: iZoSGA}}. Observe that choosing by
$\mu = \mathcal{O}\big(1/
\sqrt{(M T)}\big)$ in Algorithm \textnormal{\ref{Algorithm: iZoSGA}}, Theorem \textnormal{\ref{thm: convergence analysis}}
yields that 
$$\mathbb{E}\big\{\big\Vert \nabla \varphi^{1/(2\rho)}\big(\bm{\theta}_{t^*}\big)\big\Vert
     \big\} \leq \epsilon + \mathcal{O}\left(\sqrt{\bar{\varepsilon}}\right),$$ 
after $\mathcal{O}\left(\sqrt{S}\epsilon^{-4}\right)$ iterations. In other words, if the inexact oracle yields ``approximate solutions" to the second-stage problems \textnormal{\eqref{eqn: second-stage problem}} that are \emph{on-average} (in the sense of the definition of $\bar{\varepsilon}$) $\bar{\varepsilon}$-far from any optimal solution of \textnormal{\eqref{eqn: second-stage problem}}, then this error propagates as $\sqrt{\bar{\varepsilon}}$, when solving the two-stage problem with inexact oracles using Algorithm \textnormal{\ref{Algorithm: iZoSGA}}. 
\par Thus, our analysis characterizes the propagation of error when using inexact solution estimates of the second-stage problem (enabling the use of an arbitrary inexact oracle), and provides an insight behind the favourable numerical behaviour of \textnormal{\ref{Algorithm: iZoSGA}} demonstrated in \textnormal{\cite{IEEE_tran_Hashmietal}} in the context of IRS-assisted two-stage beamforming. Indeed, while the analysis in \textnormal{\cite{IEEE_tran_Hashmietal}} assumes exact solutions of the second-stage problem, the implemented method utilized a few iterations of WMMSE \textnormal{\cite{wmmseShi2011}} for obtaining approximate solutions instead. Nonetheless, it was numerically demonstrated  that the method performs extraordinarily well over a wide range of experiments, yielding a new state-of-the-art algorithm for passive IRS-assisted beamforming for wireless communication systems. In the following section, we will demonstrate that in this context, the magnitude of the error $\bar{\varepsilon}$ can be directly controlled by an appropriate utilization of an iterative method like WMMSE, theoretically supporing the numerical results demonstrated herein as well as in \textnormal{\cite{IEEE_tran_Hashmietal}}.

\par Before proceeding to Section \ref{sec: compatibility of assumptions}, let us briefly observe some additional insights provided by the bound in Theorem \ref{thm: convergence analysis}. Specifically, the convergence bound involves the error term $\bar{\varepsilon} = \frac{1}{T+1} \sum_{t=0}^T \mathbb{E}\left\{\varepsilon(\bm{\theta}_t,\omega_t)\right\}$. This hints that when designing practical algorithms for the solution of problems of the form of \eqref{eqn: two-stage problem}, one has great flexibility in the design of an inexact solution method for the deterministic second-stage problem. Indeed, for example, one might utilize an iterative solution method that is able to consistently yield solutions to \eqref{eqn: second-stage problem} that are at most $\bar{\varepsilon}$ away from some optimal solution, for any $\bm{\theta} \in \Theta$ and a.e. $\omega \in \Omega$. Nonetheless, it is also possible to enable larger errors during certain steps and smaller errors during others, as long as the averaged-over-the-iterates expected error, i.e. $\bar{\varepsilon}$, remains controlled. This is especially important in practical scenarios, since we may carefully adjust the quality of the inexact oracle to the needs to the optimization instance at hand with the goal of minimizing energy consumption, while expecting a smooth behaviour from the algorithm. The importance of this observation is clearly numerically demonstrated in Section \ref{sec: numerical results}-\ref{subsec: varying WMMSE iterations}.
\vspace{-4bp}
\section{Controlling the inexact oracle error} \label{sec: compatibility of assumptions}
\par At first glance, controlling the error of the inexact oracle provided in Definition \ref{definition: oracle assumption} might seem like a strong assumption, in light of the inherent non-concavity of the second-stage problem \eqref{eqn: second-stage problem}. Nonetheless, in this section we showcase that for a wide-range of problems of interest, fully compatible with Assumption \ref{assumption: two-stage problem}, we can directly control this oracle error to a reasonable extent. Indeed, controlling the random error function appearing in \eqref{eqn: solution distance error} turns out to be very natural in the setting of beamforming for wireless communication systems.
\par To clarify this, we first focus on the five conditions outlined in Assumption \ref{assumption: two-stage problem}. As thoroughly explained  in \cite[Section IV and Appendix B]{IEEE_tran_Hashmietal}, the only condition that requires verification, in the context of practical passive IRS-aided beamforming optimization, is \textbf{(A4)}. Specifically, in \cite[Appendix B]{IEEE_tran_Hashmietal} it was identified that this condition holds in three general situations. 
\par In particular, \textbf{(A4)} holds (trivially) if \eqref{eqn: second-stage problem} admits a unique solution for each $\bm{\theta} \in \Theta$ and a.e. $\omega \in \Omega$ (which is the most standard assumption utilized in the relevant literature of two-stage and bilevel optimization; e.g., see \cite{DanskinMinMax,Kwon_etal_ICML23}), or if the solution set is connected and the second-stage problem satisfies some additional regularity conditions (e.g., see \cite[Section 4]{MathOR:Shapiro}).  However, another very general setting was identified, under which \textbf{(A4)} holds. This setting is particularly interesting, since it is very natural and is met in most practical applications related to passive IRS-aided beamforming in wireless communication systems. Specifically, it was identified that if for any $\bm{\theta} \in \Theta$ and a.e. $\omega \in \Omega$ the function $\max_{\bm{W} \in \mathcal{W}} F(\bm{W},\bm{H}(\bm{\theta},\omega)) - F(\bm{W},\bm{H}(\bm{\theta},\omega))$ is \textit{sub-analytic}, then it satisfies the \emph{\L{}ojasiewicz inequality} (see \cite{LojasiewiczInequality}) with uniform exponent (see \cite[Theorem 2.3]{denkowska2018upc} or \cite{LojasiewiczInequality2}), i.e.  there exists $\eta > 0$, and for a.e. $\omega \in \Omega$ and $\bm{\theta} \in \Theta$, a positive subanalytic function $C(\bm{H}(\cdot,\cdot)) > 0$ (where $\bm{H}(\cdot,\cdot)$ is the function representing the communication channel), such that for all $\bm{W} \in \mathcal{W}$,
    \begin{align} \label{eqn: Lojasiewicz}
    %\begin{split}
        &\textnormal{dist}\Big(\bm{W}, \arg\max_{\bm{W} \in \mathcal{W}} F\left(\bm{W},\bm{H}(\bm{\theta},\omega)\right) \hspace{-2pt}
        \Big) 
        \\ & \ \leq C(\bm{H}) \Big(\max_{\bm{W} \in \mathcal{W}} F(\bm{W},\bm{H}(\bm{\theta},\omega)) - F(\bm{W},\bm{H}(\bm{\theta},\omega))\Big)^{\eta}. \nonumber
   % \end{split}
    \end{align}
\noindent It was then shown \cite[Theorem 11]{IEEE_tran_Hashmietal} that the mild condition of uniform boundedness of $C(\bm{H})$ over the image of $\bm{H}(\cdot,\cdot)$ (which can be bounded by a compact set), along with the assumption that the second-stage problem satisfies, for every $\bm{\theta} \in \Theta$ and a.e. $\omega \in \Omega$, the strong second-order sufficient optimality conditions (which require invertibility of the Lagrangian associated with \eqref{eqn: second-stage problem} evaluated at the optimal solution; see \cite[Appendix B]{IEEE_tran_Hashmietal} for additional details),
suffice to verify condition \textbf{(A4)} of Assumption \ref{assumption: two-stage problem}.
\par Note that the requirement that the second-stage problem satisfies the strong second-order sufficient optimality conditions is not particularly restrictive and, to the best of our knowledge, cannot be relaxed any further (indeed, strong second-order sufficient conditions of the second-stage problem enable the application of the \emph{implicit function theorem}, which is required to deduce any meaningful continuity or local differentiability properties of the function $\max_{\bm{W} \in \mathcal{W}} F(\bm{W},\bm{H}(\cdot,\omega))$; see, for example, the discussion in \cite{MathOR:Shapiro}). Additionally, sub-analyticity of the function  $\max_{\bm{W} \in \mathcal{W}} F(\bm{W},\bm{H}(\bm{\theta},\omega)) - F(\bm{W},\bm{H}(\bm{\theta},\omega))$ is immediate under the assumption that the utility function $F(\cdot,\cdot)$ is real-analytic (see \cite[Example 4]{SIAMOpt:Daniilidis}). We note that most utility functions employed in the relevant literature of IRS-aided optimal beamforming over wireless communication systems are indeed real-analytic (e.g., this holds for the popular sumrate utility, the proportional fairness utility as well as the harmonic-rate utility; see \cite{Prop_fairness}). 
\par In this work, we are attempting to theoretically explain the efficient algorithmic behaviour demonstrated in a plethora of experiments in \cite[Section V]{IEEE_tran_Hashmietal}. The experiments in \cite{IEEE_tran_Hashmietal} were conducted using the sumrate utility, which is a real-analytic function. Additionally, a part of the strong second-order sufficient optimality conditions (specifically, strict complementarity) was readily satisfied for the associated second-stage problems. Thus, uniform boundedness of the sub-analytic function $C(\bm{H})$ appearing in \eqref{eqn: Lojasiewicz} sufficed to prove \textbf{(A4)}. Let us then assume that there exists a constant $C_{\bm{H}} \geq C(\bm{H}(\bm{\theta},\omega))$, for any $\bm{\theta} \in \Theta$ and a.e. $\omega \in \Omega$. We now proceed to explain that if \eqref{eqn: Lojasiewicz} holds, then the magnitude error function given in Definition \ref{definition: oracle assumption} can be directly controlled. 
\par Indeed, since the second-stage optimization problem is nonconcave, we can only hope to find an approximate solution in the following sense: we find $\widetilde{\bm{W}}$ such that
\begin{equation} \label{eqn: natural error criterion}
\max_{\bm{W} \in \mathcal{W}} F(\bm{W},\bm{H}(\bm{\theta},\omega)) - F\big(\widetilde{\bm{W}},\bm{H}(\bm{\theta},\omega)\big) \leq \widetilde{\varepsilon}(\bm{\theta},\omega),
\end{equation}
\noindent in which case the magnitude of $\widetilde{\varepsilon}(\bm{\theta},\omega)$ (either in expectation or uniformly) can be controlled to some reasonable extent. However, since \eqref{eqn: Lojasiewicz} holds (under the aforementioned minor assumptions), it follows that
\[ \textnormal{dist}\Big(\widetilde{\bm{W}}, \arg\max_{\bm{W} \in \mathcal{W}} F\left(\bm{W},\bm{H}(\bm{\theta},\omega)\right)\hspace{-2pt}\Big) \leq C_{\bm{H}} \left(\widetilde{\epsilon}(\bm{\theta},\omega)\right)^{\eta},\]
\noindent for all $\bm{\theta} \in \Theta$ and a.e. $\omega \in \Omega$. In other words, there exist positive constants $C_{\bm{H}}$ and $\eta$, such that for any $\bm{\theta} \in \Theta$ and a.e. $\omega \in \Omega$, we can find an optimal solution $\bm{W}^* \in \arg\max_{\bm{W} \in \mathcal{W}} F(\bm{W},\bm{H}(\bm{\theta},\omega))$, satisfying
\[ \big\|\widetilde{\bm{W}} - \bm{W}^* \big\| \leq C_{\bm{H}} \widetilde{\varepsilon}(\bm{\theta},\omega)^{\eta}. \]
\noindent Hence, the error criterion appearing in Definition \ref{definition: oracle assumption} is in fact satisfied for some random error function $\varepsilon(\bm{\theta},\omega)$ that depends directly on the natural error criterion given in \eqref{eqn: natural error criterion}. This verifies the numerical evidence demonstrated in \cite[Section V]{IEEE_tran_Hashmietal} (and not supported in principle by the corresponding developed theory in \cite{IEEE_tran_Hashmietal}), essentially showcasing the efficiency of \ref{Algorithm: iZoSGA} developed herein, despite obtaining only approximate solutions to the second-stage problem \eqref{eqn: second-stage problem} via utilizing the well-known WMMSE algorithm \cite{wmmseShi2011}.
\par To close this section, we demonstrate the convergence rate of Algorithm \ref{Algorithm: iZoSGA}, under the assumptions specified above. As above, let us assume that there exist some uniform constants $C_{\bm{H}}$ and $\eta$, using which \eqref{eqn: Lojasiewicz} holds for any $\bm{\theta} \in \Theta$ and a.e. $\omega \in \Omega$. Let us also assume that the  (natural) error criterion given in \eqref{eqn: natural error criterion} holds, such that $\mathbb{E}\{\widetilde{\varepsilon}(\bm{\theta},\omega)^{\eta}\} < +\infty$, for all $\bm{\theta} \in \Theta$. Then, choosing 
$\mu = \mathcal{O}\big(1/\sqrt{(M T)}\big)$ yields that  
$$\scalemath{0.96}{\mathbb{E}\big\{\big\Vert \nabla \varphi^{1/(2\rho)}\big(\bm{\theta}_{t^*}\big)\big\Vert
     \big\} \leq \epsilon + \mathcal{O}\left(\sqrt{\frac{1}{T+1}\sum_{t=0}^T\mathbb{E}\left\{\widetilde{\varepsilon}(\bm{\theta}_t,\omega_t)^{\eta}\right\}}\right)},$$ 
after $\mathcal{O}\left(\sqrt{S}\epsilon^{-4}\right)$ iterations, where $\{(\bm{\theta}_t,\omega_t)\}_{t=0}^T$ are generated by \ref{Algorithm: iZoSGA}. 
\par This rate indicates that the propagation of the errors occurring by solving the second-stage problem inexactly (by following the natural criterion \eqref{eqn: natural error criterion}) potentially depend on the uniform \L{}ojasiewicz exponent $\eta$. 
Note that the existence of such a uniform exponent is guaranteed by the mere assumption that the employed utility function $F(\cdot,\cdot)$ is real-analytic (see \cite[Example 4]{SIAMOpt:Daniilidis} and \cite[Theorem 2.3]{denkowska2018upc}). 
The larger the exponent constant is, the better the error propagation. Although we do not have the ability to obtain estimates on the value of $\eta$, we observed a mild error propagation in the numerical experiments presented in \cite[Section V]{IEEE_tran_Hashmietal}. To further illustrate this, in the following section we apply \ref{Algorithm: iZoSGA} in two general, large-scale passive IRS-aided sumrate maximization settings, 
%with decreasing number of WMMSE iterations for the approximate solution of the associated second-stage problems, 
to witness the potential drop in performance relative to the degradation of the second-stage approximate solutions (i.e., as the quality of our inexact oracle decreases).

\section{Numerical Results} \label{sec: numerical results}
% Channel model + environment setup (add a graphic)
\begin{figure}
  \centering
  \centerline{\includegraphics[width=3.2in]{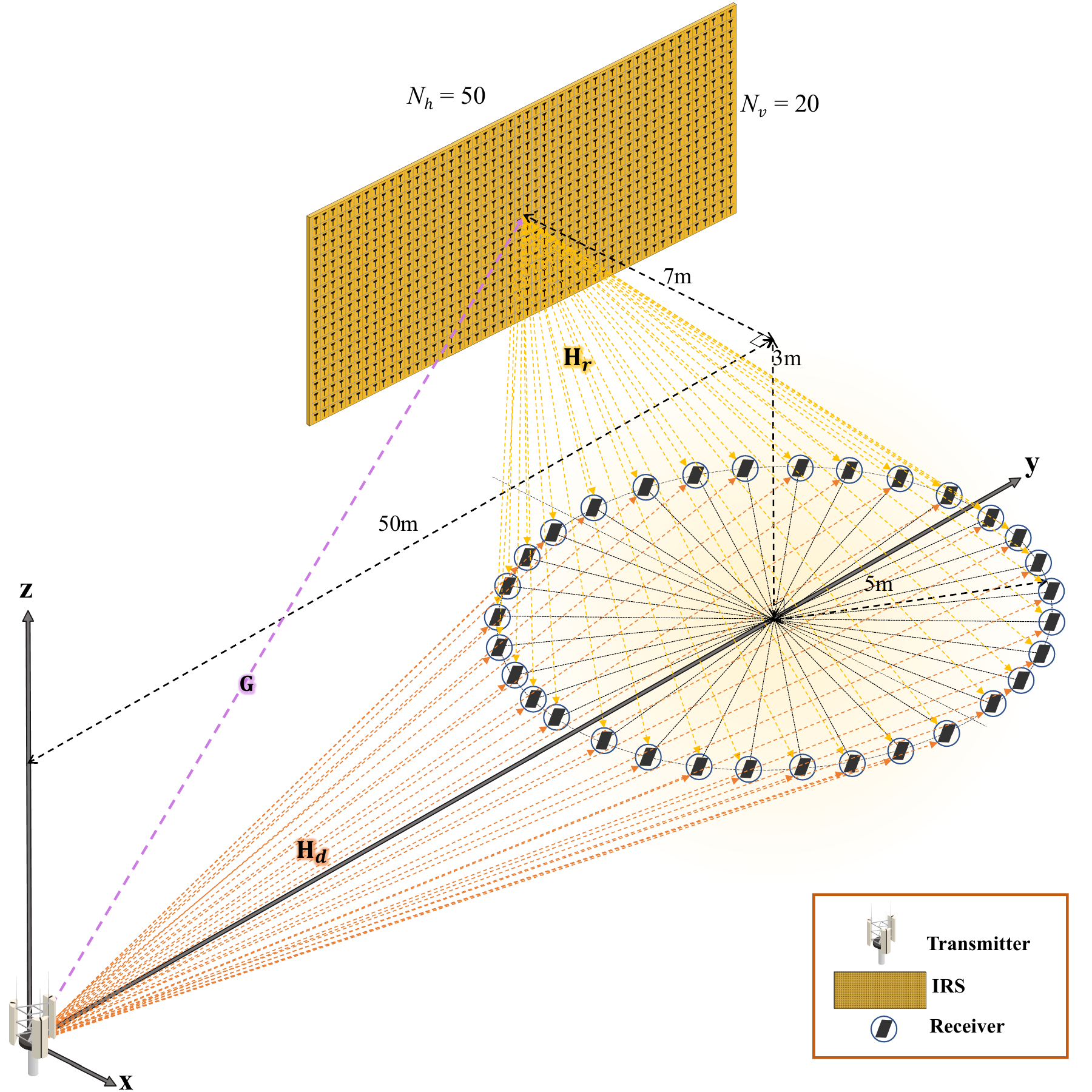}}

\caption{1000 element IRS-aided network configuration.}
\label{fig:env_setup}
\vspace{-12pt}
\end{figure}

% explain the figures in separate subsections with appropriate comments. (add a description of the hyperparameters first)

\par We assume a standard MISO downlink scenario as in Fig. \ref{fig:env_setup}, and consider a weighted sumrate utility function reading
\begin{equation}\nonumber
\scalemath{0.97}{{F}_S(\bm{W},\bm{H}(\bm{\theta},\omega)) \triangleq \sum_{k=1}^K \alpha_k \log_2\left(1 + \textnormal{SINR}_k\left(\bm{W},\bm{h}_k\left(\bm{\theta},\omega  \right)\right)\right)},
\end{equation}
where $\alpha_{k},\ \forall\ k \in \mathbb{N}^+_{K}$  are weighting scalars for each receiver $k$, and the term  $\textnormal{SINR}_k\left(\bm{W},\bm{h}_k\left(\bm{\theta},\omega  \right)\right)$ denotes the per-receiver signal-to-noise and interference ratio given by 
\begin{align}\nonumber
\begin{aligned} \label{eq:2}
\text{SINR}_k(\bm{W},\bm{h}_k(\bm{\theta}, \omega )) \triangleq \frac{\left|\bm{h}_k^\hermtr(\bm{\theta}, \omega)\bm{w}_k\right|^2}{{\sum}_{{j \in \mathbb{N}_{K}^+ \setminus k}}\left|\bm{h}^\hermtr_k(\bm{\theta}, \omega)\bm{w}_j\right|^2 + \sigma^2_k}.
\end{aligned}
\end{align}
In both expressions above, we denote the precoding vectors by $\bm{W} {=} \mathrm{vec}(\begin{bmatrix} \bm{w}_1 & \bm{w}_2 & \cdots & \bm{w}_K \end{bmatrix})\ {\in}\ \mathbb{C}^{M_U\triangleq (M \times K)}$, and the channel matrix by $\bm{H}(\bm{\theta},\omega) \triangleq \mathrm{vec}(\begin{bmatrix} \bm{h}_1 \ldots \bm{h}_K \end{bmatrix})\ {\in}\ \mathbb{C}^{M_U}$, respectively. The positive integer $M$ denotes the number of antennas at the transmitter and $K$ is the number of receivers. Finally, $\sigma^2_{k}$ is the receiver-specific noise variance.
\par We consider a power budget of $P>0$ at the transmitter which renders the feasible set $\mathcal{W}$ to that of all precoders $\bm{W}$ such that $\|\bm{W}\|^2 \leq P$. Thus, \eqref{eqn: two-stage problem} in this case takes the form  
\begin{equation} \label{eqn: two-stage sumrate}
    \max_{\boldsymbol{\theta} \in \Theta}  \mathbb{E}\left\{\max_{\bm{W}: \|\bm{W}\|^2 \leq P}  {F}_S(\bm{W},\bm{H}(\bm{\theta},\omega)) \right\}.
\end{equation}
Given this specific problem, we now proceed to specify a channel model for $\bm{H}(\bm{\theta},\omega)$, and also characterize the involved randomness (denoted by $\omega$). We assume a transmitter with $6$ antennas (i.e., $M = 6$), $32$ receivers (i.e., $K=32$), and an IRS with $1000$ phase-shifting elements (denoted by $\bm{\theta} \in \mathbb{R}^S$, with $S = 1000$). As seen in Fig. \ref{fig:env_setup}, the receivers are associated with two types of links connecting each of them to the transmitter, namely the direct LoS link $\bm{h}_{d,k}$ and the indirect non-LoS link $\bm{h}_{r,k}$ directed by the IRS towards the receivers. The link connecting the IRS to the transmitter is denoted by the channel matrix $\bm{G}$. It should be made explicit here that the randomness can be characterized by $\omega=\{\bm{G}, \bm{h}_{r,k},\bm{h}_{d,k}, k\in \mathbb{N}_K^+\}$, also noting that in Fig. \ref{fig:env_setup}, $\bm{H}_d$ and $\bm{H}_r$ denote $\{\bm{h}_{d,k}, k\in \mathbb{N}_K^+\}$ and $\{\bm{h}_{r,k}, k\in \mathbb{N}_K^+\}$, respectively. 
\par In the following experiments, these intermediate channels are modeled as Rician fading along with certain path loss gains. The exact specifications of the channel model, such as spatial correlations, statistical and instantaneous channel components along with their respective Rician factors, and path loss gains, have been reused from our previous work \cite[Section V-A]{IEEE_tran_Hashmietal} for consistency. Given the particular network topology, the channel link vector $\bm{h}_k(\bm{\theta},\omega)$ from the transmitter to the receiver $k$ may be espressed as 
\begin{align*}
    \bm{h}_k(\bm{\theta}, \omega)
    \triangleq
    \underbrace{\bm{G}^\hermtr \text{Diag}(\bm{A} \circ e^{j{\bm{\phi}}}) \bm{h}_{r,k}}_{\text{non-LOS link}} 
    + 
    \underbrace{\bm{h}_{d,k} }_{\text{LOS link}} ,
    \end{align*}
where the IRS parameters $\bm{\theta}$ are represented by amplitude and phase vectors $\bm{A} \in [0,1]^{S}$ and $\bm{\phi} \in [-2\pi,2\pi]^{S}$, respectively \cite{ch_model:wu2019intelligent}, and where $S=1000$ as mentioned earlier. The experiments that follow assume that the receivers have low mobility and thus the statistical channel state information (S-CSI) of the intermediate channels in $\omega$ has been assumed to be fixed (though unknown) throughout the service time of the transmitter. Also, in all subsequent experiments, we employ the WMMSE algorithm to perform short-term beamforming. Therefore, the oracle error $\varepsilon(\bm{\theta},\omega)$ (given in Definition \ref{definition: oracle assumption}) can be controlled solely by the number of WMMSE iterations, providing a clear interpretation of the results. Lastly, unless otherwise stated, every experimental result is averaged over 60 independent simulations.

\vspace{-4pt}
\subsection{iZoSGA  with fixed WMMSE iterations} \label{subsec: fixed WMMSE iterations experiments}
\par In the first set of experiments we run the WMMSE algorithm for different number of iterations, as shown in Fig. \ref{fig:1}, while also fine-tuning the IRS parameters with iZoSGA using the resulting precoders having varying levels of \textit{inexactness}, noting that the precoder (i.e., oracle) errors depend on the number of WMMSE iterations used. As shown in Fig. \ref{fig:1}, the beamformers optimized with 10, 20 and 50 WMMSE iterations attain similar performance. Indeed, we observe that after 10 WMMSE iterations, the reduction in the magnitude of the oracle error $\varepsilon(\bm{\theta},\omega)$ is not significant, resulting in negligible performance gains. On the other hand, we see that iterating WMMSE for 1, 2 or 3 iterations and optimizing iZoSGA with the corresponding precoders results in noticeably worse performance (albeit with significant gains as compared with the random IRS parameter setting). In these cases, we observe that the resulting oracle error magnitudes, $\varepsilon(\bm{\theta},\omega)$, have a substantial effect in the final QoS; this is consistent with the conclusions presented in the theoretical results of Section \ref{sec: conv anal}. 

\begin{figure}[t]
\centering
\centerline{\includegraphics[width=\linewidth]{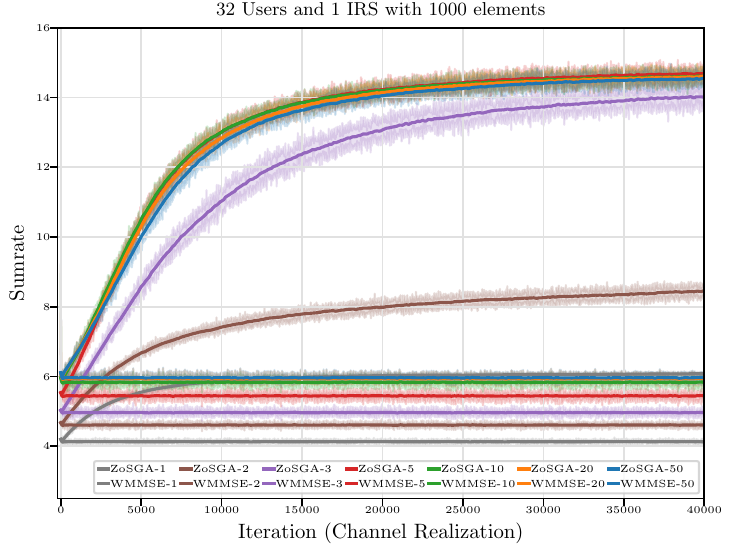}}
\vspace{-6bp}
\caption{Average sumrates achieved by WMMSE \cite{wmmseShi2011} (random IRS phase-shifts) and iZoSGA, with a 1000 phase-shifter IRS using 1,2,3,5,10,20 and 50 WMMSE iterations.}
\label{fig:1}
\vspace{-15pt}
\end{figure}

\par A more interesting behavior is observed when the WMMSE algorithm is run for 5 iterations. There is a discernible performance gap between WMMSE with 5 and 10 iterations when the IRS parameters are tuned randomly. However, by employing iZoSGA to fine-tune the IRS parameters and fixing the WMMSE iterations merely to 5, we were able to obtain comparable performance with the algorithm that enabled 50 WMMSE iterations. Two relevant observations are in order. Firstly, within the first 2500 (outer) iterations, iZoSGA fine-tunes the IRS parameters to match the performance of iZoSGA with 50 (or, in general, more than 10) WMMSE iterations. It seems that the zeroth-order sample gradients, although more erratic, still provide enough information to optimize the IRS parameters sufficiently, so as to counteract larger precorder errors. Thus, the IRS parameters in a sense \textit{fill in the gaps} and we see that the learning curves reach competitive performance levels, as long as WMMSE is allowed to run for at least 5 iterations. Secondly, we observe that after reaching the level of the learning curves pertaining to higher WMMSE iterations, the learning curve for the 5-iteration case follows almost the same trajectory, thus establishing that iZoSGA with 5-iteration WMMSE oracle was able to quickly match the steeper climb of the methods that use higher-quality oracles observed during initial iZoSGA iterations. At this point we should note, in light of our previous work \cite[Section V]{IEEE_tran_Hashmietal}, that the performance gains achieved by the iZoSGA algorithm are state-of-the-art for the IRS-aided wireless network instances at hand.

\par From the experiments above, we conclude that by controlling the error $\varepsilon(\bm{\theta},\omega)$ below a certain threshold ---in this network setting, this was achieved by running WMMSE for 5 iterations--- iZoSGA can ``optimally" tune the IRS parameters to achieve competitive performance, irrespectively of the level of inexactness beyond this threshold point. This is a useful empirical observation that could further impact algorithmic design in practical settings so as to reduce energy consumption.

\vspace{-5pt}
\subsection{iZoSGA with varying WMMSE iterations} \label{subsec: varying WMMSE iterations}

\par To evaluate the effects of substantial changes in the orcale error $\varepsilon(\bm{\theta},\omega)$ during the learning process on the convergence of iZoSGA, we perform two experiments in which we change the WMMSE iterations after every 8000 outer iterations, as shown in Fig. \ref{fig:2}. We start the first experiment with 20 WMMSE iterations, which are subsequantly decreased to 10, 7, 6 and 5 at the 8000, 16000, 24000 and 32000 outer iteration marks, respectively. As is visible in Fig. \ref{fig:2}, the orange curve converged to the best achievable weighted sumrate level without any abrupt changes, verifying that as long as the error $\varepsilon(\bm{\theta},\omega)$ is below a certain threshold, say $\varepsilon'$, employing the iZoSGA algorithm leads to optimal performance, verifying further our observations from Section \ref{sec: numerical results}-\ref{subsec: fixed WMMSE iterations experiments}. In other words, if we were to define an error bound stopping criterion for the WMMSE algorithm (as proposed in the original work \cite{wmmseShi2011}), we would still be able to obtain consistent results as long as the error bound is less than some appropriate threshold $\varepsilon'$.

\begin{figure}[t]
\centering
\centerline{\includegraphics[width=\linewidth]{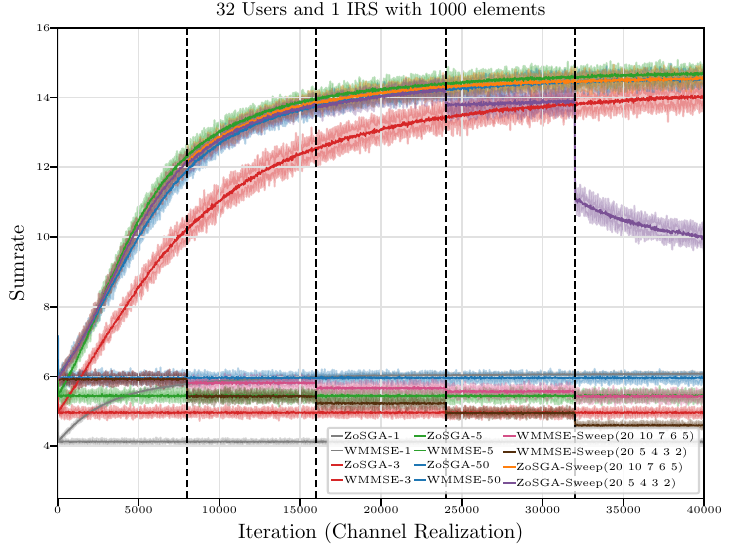}}
\vspace{-6bp}
\caption{Average sumrates achieved by WMMSE \cite{wmmseShi2011} (random IRS phase-shifts) and two separate iZoSGA experiments with WMMSE iterations decreasing after every $8000$ iterations ($20 {\to} 10 {\to} 7 {\to} 6 {\to} 5$ and $20 {\to} 5 {\to} 4 {\to} 3 {\to} 2$).}
\label{fig:2}
\vspace{-12pt}
\end{figure}

\par Contrary to the above, if at any point it happens that $\varepsilon(\bm{\theta},\omega) > \varepsilon'$, then we will observe a performance degradation of iZoSGA. The second experiment confirms this hypothesis. Once again, we start by employing WMMSE for 20 iterations but now we decrease the iterations to 5, 4, 3 and 2 as iZoSGA progresses. As long as WMMSE is run for at least 5 iterations, i.e. up until 16000 (outer) iZoSGA iterations, the purple curve is consistent with the orange one. However, when the WMMSE iterations are decreased to 4, we see a small drop in the weighted sumrates indicating that the error has now become greater than the required threshold $\varepsilon'$. At 3 iterations, we observe a noticeable decrease in the weighted sumrates, at which point the curve flattens out, indicating that iZoSGA is not able to fine tune IRS parameters any further with the corresponding level of inexactness. Finally, when WMMSE iterations are decreased to 2, we not only observe that the weighted sumrates drops significantly, but also that the learning curve descends. This behavior indicates that the increase in the oracle errors not only decreases the level of performance, but also \textit{corrupts} learning by overwriting the already tuned IRS parameters with more inexact gradient steps. 
\par Once again, we empirically verify that as long as oracle errors remain above a certain threshold (irrespectively of how far from this threshold these errors are), the performance of iZoSGA is not hindered. At the same time, iZoSGA is adaptive to fluctuations in $\varepsilon(\bm{\theta},\omega)$; this is already suggested from our theoretical error bound, which is \emph{ergodic} in nature.

\begin{figure}[t]
\centering
\centerline{\includegraphics[width=\linewidth]{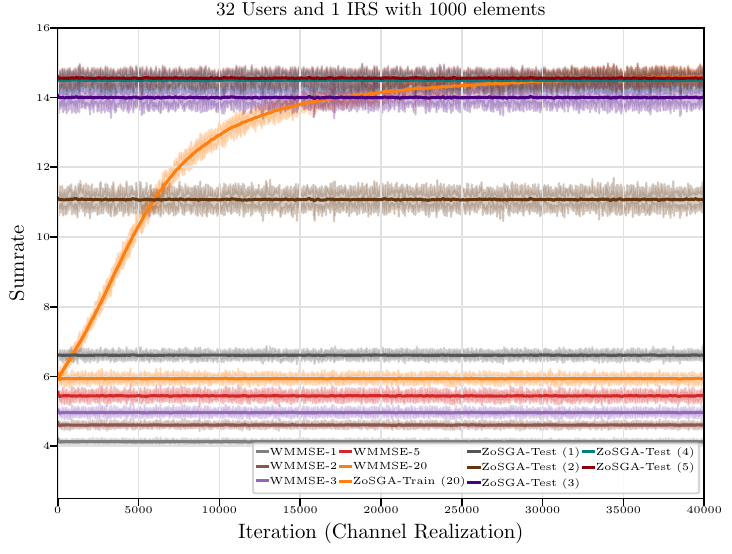}}
\vspace{-6bp}
\caption{Average sumrates achieved by WMMSE \cite{wmmseShi2011} (at 1,2,3,4 and 5 iterations) with iZoSGA fine-tuned 1000 IRS phase-shifters.}
\label{fig:3}
\vspace{-13pt}
\end{figure}

\vspace{-4pt}
\subsection{Tuned IRS tested with inexact WMMSE precodings}
\par As shown in Fig. \ref{fig:2}, when IRS parameters are tuned (using iZoSGA) with a high-quality \textit{initial} inexact oracle (i.e., in our case by using a large number of WMMSE iterations), decreasing the oracle quality considerably at later iterations (up to a certain level) does not degrade the achieved performance. This motivated us to perform a set of experiments where the IRS was first tuned (i.e., \emph{trained}) with a sufficiently large number of WMMSE iterations, and then subsequently \emph{tested} (or deployed) using a lower-quality oracle (i.e., less WMMSE iterations, and thus computationally cheaper short-term precoders) and the \emph{constant} but tuned (or learnt) IRS parameters, to observe the overall performance of the algorithm.  
\par We first trained the IRS by enabling 20 WMMSE iterations. After the IRS parameter tuning is complete, which in essence results in a new network setting (via channel shaping), we run the WMMSE algorithm with 1, 2, 3, 4 and 5 iterations, respectively, as shown in Fig. \ref{fig:3}. We observe that using just 4 or 5 WMMSE iterations, the resulting beamforming performance is on par with the 20-WMMSE iterations regime (with active IRS tuning). In other words, we confirm that one can first train iZoSGA with more WMMSE iterations (i.e., using a high-quality oracle) and can then deploy the fine tuned IRS with a more realistic compute overhead of just 4 WMMSE iterations, without any discernible performance degradation. 

\par Some interesting observations can also be drawn from the 1-, 2- and 3-WMMSE iteration(s) settings. In the 3-WMMSE iterations case, we can see that the respective learning curve in Fig. \ref{fig:1} achieves the same performance as that in Fig. \ref{fig:3} (i.e., trained and then deployed using 3 WMMSE iterations). In the 1- and 2-WMMSE iteration(s) cases, we get something somewhat unexpected: The weighted sumrate levels corresponding to the said iterations in Fig. \ref{fig:3} are noticeably higher than those in Fig. \ref{fig:1}. This confirms the effectiveness of our \textit{train expensive, deploy cheap} approach.

\vspace{-4pt}
\subsection{iZoSGA with Physical IRS}
\par Lastly, we extended our modeling assumptions to a practically feasible tranmission line (TL) equivalent of an electromagnetic (EM) IRS model as proposed in \cite{phy_em_model:costa2021electromagnetic}, based on varactor diodes. This TL model incorporates the various geometrical and electrical properties of the IRS elements, such as changes to the responses due to different angles of signal wave incidence, mutual coupling among closely spaced elements, reflection losses, etc. For a more detailed description of the TL model, see \cite[Section V-C]{IEEE_tran_Hashmietal}. Suffice to say that each element $\theta(C_{var})$ of the IRS parameter vector $\bm{\theta}$ is now a function of the varactor diode capacitance $C_{var}$. By assuming that the Floquet theorem holds \cite{phy_em_model:costa2021electromagnetic}, this function relation is the same for every element's capacitance. Thus, we may propagate the above relation for all IRS elements, say $S$ in number, and define a vector function $\bm{\theta}(\cdot)$ of the varactor capacitances $\bm{c}_{var} =  [C^{1}_{var} \, C^2_{var} \, C^{3}_{var} \, \cdots \, C^S_{var}]^\top$ collectively expressed as $\bm{\theta}(\bm{c}_{var}) = [ \theta(C^{1}_{var}) \, \theta(C^2_{var}) \, \theta(C^{3}_{var}) \, \cdots \, \theta(C^S_{var})]^\top$.  Thus, the channel link vector $\bm{h}_k(\bm{\theta},\omega)$ from the transmitter to the receiver $k$ can now be rewritten as
\begin{align*}
    \bm{h}_k(\bm{\theta}, \omega)
    \triangleq
    \bm{G}^\hermtr \text{Diag}(\bm{\theta}(\bm{c}_{var})) \bm{h}_{r,k} 
    + \bm{h}_{d,k}  ,
    \end{align*}
where, instead of tuning phases and amplitudes directly, we fine-tune the capacitance vector $\bm{c}_{var}$. 

\begin{figure}[t]
\centering
\centerline{\includegraphics[width=\linewidth]{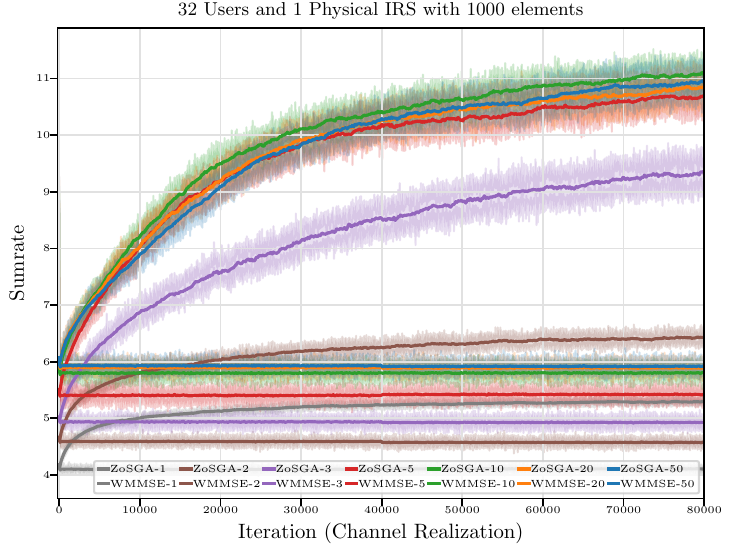}}
\vspace{-6bp}
\caption{Average sumrates achieved by WMMSE \cite{wmmseShi2011} (random IRS phase-shifts) and iZoSGA with a physical EM IRS
model with 1000 phase-shifters using 1,2,3,5,10,20 and 50 WMMSE iterations.}
\label{fig:4}
\vspace{-13pt}
\end{figure}

\par By attempting to replicate the first set of experiments (as reported in Fig. \ref{fig:1}), we conduct a series of runs of the algorithm with different (constant) WMMSE iterations, as shown in Fig. \ref{fig:4}.  We first observe that the learning curves have higher variances compared to those shown in Fig. \ref{fig:1} (over a set of 60 independent simulations). This is because the amplitude/phase response of each varactor diode is constrained relative to its capacitance values, which is not the case for the ideal IRS set-up \cite[Fig. 7]{phy_em_model:costa2021electromagnetic}. Nonetheless, the trend that we observed in the previous subsections is consistent with the one shown in Fig. \ref{fig:4}, i.e., in the 1-, 2- and 3-WMMSE iteration regimes, iZoSGA achieves visibly worse performance levels when compared with the 5-, 10-, 20- and 50-WMMSE iteration regimes. Moreover, due to an increase in the problem complexity, all the learning curves require more iterations to converge (i.e., to flatten).  
%No need for this:We believe that by increasing the number of independent simulations, we may reduce the observed variances. However, in light of our limited computational resources, we have decided to consistently report the averaged results over 60 simulations. 
It is also important to observe in this case the degradation in the overall performance, compared with the performance reported for the idealized IRS setting of Sections \ref{sec: numerical results}-\ref{subsec: fixed WMMSE iterations experiments}, which is to be expected as the varactor structure couples the phase and amplitude responses, effectively restricting the parameter space $\Theta$ available to iZoSGA.

\vspace{-4pt}
\section{Conclusions} \label{sec: Conclusions}
In this paper, we developed a fully data-driven and model-free zeroth-order algorithm suitable for the solution of two-stage stochastic problems arising in passive IRS-assisted beamforming over wireless communication networks. The proposed method utilizes inexact oracles for short-term precoding, and zeroth-order sample-gradient steps for long-term tuning of the IRS parameters. We have shown, under very general and realistic assumptions, that iZoSGA converges close to a stationary point of the underlying (nonconvex) two-stage stochastic program, where the proximity to this stationary point is directly dependent on the magnitude of the error of the associated inexact oracle. By specializing our analysis to a wide range of passive IRS-assisted optimal beamforming problems, we showcase that the inexact oracle error can be controlled to a reasonable extent. In turn, we obtain valuable insights that enable the design of practical and highly effective algorithmic schemes able to solve realistic (large-scale) instances, without structural assumptions on channel models, network topology or IRS configurations, and with minimal effort requirements.

The effectiveness and reliability of the proposed approach were numerically demonstrated on a wide range of meaningful settings, using both idealized and (realistic) physical IRS models, assuming a network topology with a large number of users. The state-of-the-art developments in this paper could pave the way for the seamless utilization of oracle-based zeroth-order algorithms, such as iZoSGA, in the context of IRS-aided optimal beamforming, by bridging the gap between theory and practice while also providing insights for appropriate design and deployment of such systems. 
%\vfill\pagebreak
\bibliographystyle{IEEEbib-abbrev}
\bibliography{iZoSGA_references.bib}

\end{document}